\documentclass[a4paper]{lipics-v2021}

\usepackage{amsmath,amssymb,amsthm}
\usepackage{graphicx}
\usepackage{bigstrut}

\usepackage{bm}

\makeatletter
\newtheorem*{rep@theorem}{\rep@title}
\newcommand{\newreptheorem}[2]{%
	\newenvironment{rep#1}[1]{%
		\def\rep@title{#2 \ref{##1}}%
		\begin{rep@theorem}}%
		{\end{rep@theorem}}}
\makeatother

\newenvironment{lemma-repeat}[1]{\begin{trivlist}
		\item[\hspace{\labelsep}{\bf\noindent Lemma \ref{#1}.}]\em }%
	{\end{trivlist}}
\newenvironment{theorem-repeat}[1]{\begin{trivlist}
		\item[\hspace{\labelsep}{\bf\noindent Theorem \ref{#1}.}]\em }%
	{\end{trivlist}}
\newenvironment{claim-repeat}[1]{\begin{trivlist}
		\item[\hspace{\labelsep}{\bf\noindent Claim \ref{#1}.}]\em }%
	{\end{trivlist}}
	\newenvironment{proposition-repeat}[1]{\begin{trivlist}
		\item[\hspace{\labelsep}{\bf\noindent Proposition \ref{#1}.}]\em }%
	{\end{trivlist}}

\newcommand{\qedsymb}{\qed}

\def\eps{\varepsilon}
\newcommand{\setup}{\textsf{setup}}

\newcommand{\set}[1]{\left\{ #1 \right\}}

\DeclareMathOperator{\poly}{poly}

\DeclareMathOperator{\polylog}{\mathrm{polylog}}

\newcommand{\remove}[1]{}

\usepackage{xcolor}

\newcommand{\proc}[1]{\mathsf{#1}}

\renewcommand\paragraph[1]{\vspace{2mm}\noindent {\bf #1}}

\usepackage{amsmath}
\usepackage[linesnumbered,vlined]{algorithm2e}
\usepackage{comment}
\usepackage{amsthm}
\graphicspath{ {./images/} }

\usepackage{latexsym}
\usepackage{amssymb}
\usepackage{todonotes}
\usepackage{comment}
\usepackage{mathtools}
\usepackage{tikz}

\theoremstyle{remark}

\theoremstyle{definition}

\newcommand{\E}{\mathbf{E}}

\newcommand{\ip}[1]{\left}
\usepackage{xspace}
\newcommand{\congest}{\textsc{Congest}\xspace}
\newcommand{\qcongest}{\textsc{Quantum Congest}\xspace}
\newcommand{\CONGEST}{\textsc{Congest}\xspace}

\newcommand{\clique}{\textsc{Congested Clique}\xspace}
\newcommand{\local}{\textsc{LOCAL}\xspace}
\newcommand{\qclique}{\textsc{Quantum Congested Clique}\xspace}

\newcommand{\rv}[1]{\mathrm{#1}}
\newcommand{\Tri}{\textsc{FindTriangleInSubnetwork}}

\newcommand{\coloneq}{:=}
\newcommand{\hide}[1]{ }

\renewcommand{\mathbf}{\bm}
\newcommand{\var}[1]{\mathit{#1}}

\newcommand{\inext}{_\mathrm{in}}
\newcommand{\outext}{_\mathrm{out}}

\RestyleAlgo{boxruled}
\LinesNumbered
\let\oldnl\nl
\newcommand{\nonl}{\renewcommand{\nl}{\let\nl\oldnl}}

\SetKw{Send}{send}

\newcommand{\mix}{\ensuremath{\tau_{\operatorname{mix}}}}

\title{Quantum Distributed Algorithms for Detection of Cliques}

\author{Keren Censor-Hillel}{Technion, Israel}{ckeren@cs.technion.ac.il}{https://orcid.org/0000-0003-4395-5205}{}
\author{Orr Fischer}{Tel-Aviv University, Israel}{orrfischer@mail.tau.ac.il}{}{}
\author{Fran{\c c}ois Le Gall}{Nagoya University, Japan}{legall@math.nagoya-u.ac.jp}{}{}
\author{Dean Leitersdorf}{Technion, Israel}{dean.leitersdorf@gmail.com}{}{}
\author{Rotem Oshman}{Tel-Aviv University, Israel}{roshman@tau.ac.il}{}{}

\authorrunning{K. Censor-Hillel, O. Fischer, F. Le Gall, D. Leitersdorf, R. Oshman} 

\Copyright{Keren Censor-Hillel, Orr Fischer, Fran{\c c}ois Le Gall, Dean Leitersdorf, Rotem Oshman} 

\ccsdesc[500]{Networks~Network algorithms}
\ccsdesc[500]{Theory of computation~Distributed algorithms}

\keywords{distributed graph algorithms, quantum algorithms, cycles, cliques, Congested Clique, CONGEST} 

\funding{This project was partially supported by the European Union's Horizon 2020 Research and Innovation Programme under grant agreement no.~755839. Fran{\c c}ois Le Gall was supported by JSPS KAKENHI grants Nos.~JP19H04066, JP20H05966, JP20H00579, JP20H04139, JP21H04879 and MEXT Quantum Leap Flagship Program (MEXT Q-LEAP) grants No.~JPMXS0118067394 and JPMXS0120319794.}

\nolinenumbers 

\EventEditors{Mark Braverman}
\EventNoEds{1}
\EventLongTitle{13th Innovations in Theoretical Computer Science Conference (ITCS 2022)}
\EventShortTitle{ITCS 2022}
\EventAcronym{ITCS}
\EventYear{2022}
\EventDate{January 31--February 3, 2022}
\EventLocation{Berkeley, CA, USA}
\EventLogo{}
\SeriesVolume{215}
\ArticleNo{93}
\begin{document}
\maketitle

\begin{abstract}
The possibilities offered by quantum computing have drawn attention in the distributed computing community recently, with several breakthrough results showing quantum distributed algorithms that run faster than the fastest known classical counterparts, and even separations between the two models. A prime example is the result by Izumi, Le Gall, and Magniez [STACS 2020], who showed that triangle detection by quantum distributed algorithms is easier than triangle listing, while an analogous result is not known in the classical case.

In this paper we present a framework for fast quantum distributed clique detection. This improves upon the state-of-the-art for the triangle case, and is also more general, applying to larger clique sizes.

Our main technical contribution is a new approach for detecting cliques by encapsulating this as a search task for nodes that can be added to smaller cliques. To extract the best complexities out of our approach, we develop a framework for \emph{nested} distributed quantum searches, which employ checking procedures that are quantum themselves.

Moreover, we show a circuit-complexity barrier on proving a lower bound of the form $\Omega(n^{3/5+\epsilon})$ for $K_p$-detection for any $p \geq 4$, even in the classical (non-quantum) distributed CONGEST setting.
\end{abstract} 

\vspace*{\fill}
\pagebreak

\section{Introduction}
\noindent{\bf Quantum distributed computing.}
The power of quantum computing in the distributed setting has recently been the subject of intensive investigations  \cite{EKNP14,Izumi+PODC19,LeGall+PODC18,LNR19}.
The main difference between classical and quantum distributed computing is that the quantum setting,
quantum information, i.e., quantum bits (qubits), can be sent through the edges of the network instead of classical information (i.e., bits).
Le Gall and Magniez \cite{LeGall+PODC18} and Izumi and Le Gall \cite{Izumi+PODC19}, in particular, have shown the superiority of quantum distributed computing over classical distributed computing for two fundamental models, the \congest model and the \clique model.

The (classical) \congest model is one of the most studied models in classical distributed computing. In this model, $n$ nodes communicate with each other over the network by exchanging messages of $O(\log n)$ bits in synchronous rounds. All links and nodes are reliable and suffer no faults. Each node has a distinct identifier, but the network graph is not initially known to the nodes.
In the quantum version of this model (which we denote \qcongest), as defined in \cite{EKNP14,LeGall+PODC18}, the only difference is that the nodes can exchange quantum information: each message exchanged consists of $O(\log n)$ quantum bits instead of $O(\log n)$ bits in the classical case. In particular, initially the nodes of the network do not share any entanglement.
Achieving quantum speedups in this setting is especially challenging since, as shown by Elkin et al.~\cite{EKNP14}, the ability to send quantum information is not helpful for many crucial components of distributed algorithms (e.g., routing or broadcast of information). Le Gall and Magniez \cite{LeGall+PODC18} have nevertheless showed the superiority of quantum distributed computing in this model: they constructed a $\tilde{O}(\sqrt{n})$-round quantum algorithm for the exact computation of the diameter of the network (for networks with small diameter), while it is known that any classical algorithm in the \congest model requires~$\tilde{\Omega}(n)$ rounds, even for networks with constant diameter~\cite{FHW12}. 

The (classical) \clique model is similar to the (classical) \congest model, but it separates the input to the problem we are working on from the communication topology: the input is some graph $G = (V,E)$, but the communication topology allows all nodes to communicate directly with one another (i.e., a clique). This model is close in flavor to massively-parallel computing, but we focus only on rounds and communication, not memory.
The \qclique model is defined as the quantum version of the \clique model: the only difference is again that each exchanged message consists of $O(\log n)$ quantum bits instead of $O(\log n)$ bits. 
Izumi and Le Gall~\cite{Izumi+PODC19} showed that quantum distributed algorithms can be more powerful than classical distributed algorithms in the \clique model as well: they constructed a quantum algorithm faster than the best known classical algorithms for the All-Pair Shortest Path problem.

\paragraph{Distributed subgraph detection.}
In the past few years there has been a surge of works on classical distributed algorithms investigating the complexity of the subgraph detection problem, which asks to detect the existence of a specified subgraph within the input graph (which coincides with the communication network in the \congest and \qcongest models). 
For small subgraphs, the subgraph-detection problem is extremely \emph{local}, and yet it is challenging to solve in the $\congest$ model, due to the restricted bandwidth.
Most prior work has focused on detecting $p$-cliques (denoted below~$K_p$) and $\ell$-cycles (denoted below~$C_\ell$), for small values of~$p$ and~$\ell$. In particular, in the $p$-clique detection problem (also called the $K_p$-freeness problem), the goal is to decide if the input graph contains a $p$-clique or not. If it contains a $p$-clique then at least one node must output ``yes''. Otherwise all nodes must output ``no''. (A detailed review of classical algorithms for $p$-clique detection is given at the end of this section.)

Izumi et al. \cite{ILM20} recently showed that quantum algorithms can give an advantage for triangle detection (i.e., $K_3$ detection), by constructing a quantum algorithm with round complexity $\tilde O(n^{1/4})$ in the \qcongest model (the best known classical algorithm for triangle detection, by Chang and Saranurak~\cite{CS19}, has complexity $\tilde O(n^{1/3})$). 
The key technique used in \cite{ILM20} is \emph{distributed Grover search}, which was introduced in \cite{LeGall+PODC18} and consists in a distributed implementation of Grover's celebrated quantum algorithm \cite{Grover96}.

\paragraph{Our results.} In this paper we further investigate the power of quantum distributed algorithms. Using a new approach based on \emph{nested} distributed quantum searches, we obtain quantum algorithms for $p$-clique detection, in both the \qcongest model and the \qclique model, that outperform their classical counterparts.

We first consider clique detection in the \qclique model. Our results are summarized in the following theorem (the upper bounds we obtain for $K_p$-detection with $p\ge 5$ are actually even stronger --- see Corollary \ref {cor:newCCextend} in Section \ref{sec:CCfast}).
\begin{theorem}\label{th:clique-clique}
There exists a quantum algorithm that solves $p$-clique detection with success probability at least $1-1/\poly(n)$ in the \qclique model with complexity $\tilde O(n^{1/5})$ for $p=3$
and  $\tilde O(n^{1-2/(p-1)})$ for $p\ge 4$.
\end{theorem}

For all values $p\ge 3$, the quantum algorithms we obtain by Theorem \ref{th:clique-clique} are faster than all known classical and quantum algorithms for $p$-clique detection.

We then investigate clique detection in the \qcongest model.
\newcommand{\THcongest}
{
There exists a quantum algorithm that solves $p$-clique detection with success probability at least $1-1/\poly(n)$ in the \qcongest model with complexity $\tilde O(n^{1/5})$ for $p=3$,
and $\tilde O(n^{1-2/(p-1)})$ for $p\ge 7$.
}
\begin{theorem}\label{th:clique-congest}
\THcongest
\end{theorem}

For all $p=3$ and $p\ge 7$, the quantum algorithms we obtain by Theorem~\ref{th:clique-congest} are faster than all known algorithms for $p$-clique detection in the \congest or \qcongest model.
For $4 \leq p \leq 6$, our approach currently does not lead to an improvement over the classical algorithms from \cite{Censor+SODA21}, which are respectively $\tilde O(\sqrt{n})$, $\tilde O(n^{3/5})$ and $\tilde O(n^{2/3})$. 
The reason is, informally, that in the $\qcongest$ model we (and all prior work on listing cliques)
decompose the graph into well-connected clusters, and simulate the $\qclique$ on each cluster;
for very small cliques, the effort required to collect the edges needed at each cluster overwhelms any time savings
we currently gain from the quantum search.


%
%


Finally, we consider lower bounds. While a tight $\Omega(\sqrt{n})$ lower bound for $K_4$-detection is known in the (classical) $\congest$ model \cite{CK20}, tight lower bounds are not known for larger cliques,
and there is no known non-trivial lower bound for $K_p$-detection in the quantum setting.
We show a barrier for proving an $\Omega(n^{3/5+\epsilon})$ lower bound for $p$-clique detection for any $\eps > 0$. Namely, such a bound would imply breakthrough results in the field of circuit complexity, which are far beyond the current state-of-the-art. We actually show this barrier for the (classical) $\congest$ model, but since any lower bound in the $\qcongest$ model also holds the $\congest$ model, this barrier holds for the $\qcongest$ model as well.\footnote{Note that the statement of Theorem \ref{thm:Kp_obstacle} is actually interesting only for $p\ge 6$, since for $p=4,5$ we know that there exist algorithms beating this barrier: as already mentioned, algorithms with complexity $\tilde O(\sqrt{n})$ for $p=4$ and $\tilde O(n^{3/5})$   for $p=5$ are given in \cite{Censor+SODA21}.} 
\begin{theorem}
	 For any constant integer $p \geq 4$ and any constant $\epsilon > 0$, proving a lower bound of $\Omega(n^{3/5+\epsilon})$ on $p$-clique detection in the $\congest$ model would imply new lower bounds on high-depth circuits with constant fan-in and fan-out gates.
	\label{thm:Kp_obstacle}
\end{theorem}

Previously, several such barriers were known: in \cite{EFFKO19} it is shown that there is an absolute constant $c$ such that proving a lower bound of the form $\Omega(n^{1-c})$ for $C_{2k}$-detection would imply breaking a circuit complexity barrier. In \cite{EFFKO_note} such a barrier is shown for proving an $\Omega(n^\epsilon)$ lower bound on triangle detection, and in \cite{CHFGLGDR20} a barrier for proving a lower bound of $\Omega(n^{1/2+\epsilon})$ for $C_6$-detection is given.

We refer to Table \ref{table:results} for a summary of our results and a detailed comparison with prior works.
\vspace{2mm}
\begin{table}[tb]\centering
  \caption{Our upper bounds for subgraph detection, and the corresponding known results in the classical setting. Here $n$ denotes the number of nodes in the network. Note that in \qclique, the algorithms we obtain for $K_p$-detection with $p\ge 6$ are even faster than shown here, see Corollary \ref {cor:newCCextend}.} 
 \begin{tabular}{|p{2.3cm}|p{3.3cm}|p{6.5cm}|} \hline
    Subgraph & \clique & \qclique \\ \hline \hline
    \multirow{2}{*}{$K_3$}&\multirow{2}{*}{$\tilde{O}(n^{1/3})$\,\hspace{14mm} \cite{CS19}}& $\tilde{O}(n^{1/4})$ \bigstrut \hspace{17mm}\cite{ILM20}\\
                                       && $\tilde{O}(n^{1/5})$\hspace{16mm}  Theorem~\ref{th:clique-clique} \\ \hline
    $K_p$ ($p\ge 4$)& $\tilde{O}(n^{1-2/p})$\:\bigstrut\hspace{10mm} \cite{Censor+SODA21}& $\tilde{O}(n^{1-2/(p-1)})$  \hspace{5mm}  Theorem~\ref{th:clique-clique}\\ \hline
  \end{tabular}
\vspace{4mm}

  \begin{tabular}{|p{2.3cm}|p{3.3cm}|p{6.5cm}|} \hline
    Subgraph & \congest & \qcongest \\ \hline \hline
    \multirow{2}{*}{$K_3$}&\multirow{2}{*}{$\tilde{O}(n^{1/3})$\,\hspace{14mm}\cite{CS19}}& $\tilde{O}(n^{1/4})$\bigstrut  \hspace{17mm}\cite{ILM20}\\
                                       && $\tilde{O}(n^{1/5})$  \hspace{15mm}  Theorem~\ref{th:clique-congest} \\ \hline
    \multirow{2}{*}{$K_p$ ($p\ge 7$)}&$\tilde{O}(n^{1-2/p})$\bigstrut\hspace{11mm}\cite{Censor+SODA21}& \multirow{2}{*}{$\tilde{O}(n^{1-2/(p-1)})$  \hspace{5mm}  Theorem~\ref{th:clique-congest}}\\
                                       &${\Omega}(\sqrt{n})$\bigstrut\hspace{12mm}\,\,\cite{CK20}&  \\ \hline
  \end{tabular}\label{table:results}
\end{table}

\paragraph{Overview of our main technique.}
Our key approach is to encapsulate $K_p$-detection as a search task, and use a distributed implementation~\cite{LeGall+PODC18} of Grover search~\cite{Grover96} to solve the task.
Grover's algorithm consists of alternating between quantum operations called \emph{Grover diffusion operations},
and checking operations, also called \emph{checking queries}. The total number of operations is ${O}(\sqrt{|X|})$,
where $X$ is the search domain.
In the distributed implementation developed in \cite{LeGall+PODC18}, one specific node of the network (the \emph{leader}) executes each Grover diffusion operation locally, but implements each query in a distributed way using a distributed checking procedure. All prior works using this framework (\cite{Izumi+PODC19, LeGall+PODC18,ILM20}) considered the setting where the checking procedure is a classical procedure. In this work, we consider checking procedures that themselves also apply distributed Grover searches; we develop a framework to describe such \emph{nested} distributed quantum searches (see Lemma \ref{lemma:qsearch} in Section~\ref{section:prelim}). Our framework can actually be applied in a completely ``black-box'' way to design quantum distributed algorithms (i.e., no knowledge of quantum computation is needed to apply this framework).

The main challenge in applying quantum search to the clique-detection problem is that the search-space is very large: we must search over the $\Theta(n^p)$ possibilities, and a na\"ive approach would require $\Theta(n^{p/2})$ quantum queries, which is extremely inefficient. Instead, we show that one can carefully split the search into nested stages, so that each stage adds a single node to the clique we are trying to find. Crucially, nesting the stages of the search allows us to re-use information computed in one stage for all the search queries in the next stage: in each stage, we have already found some $\ell$-cliques, where $\ell < p$, and we want to add one more node to the cliques, to obtain $(\ell+1)$-cliques (until in the final stage we obtain $p$-cliques). To this end, the nodes collect some edges, which allow them to detect some $(\ell+1)$-cliques, and then use a nested search to try to complete the $(\ell+1)$-cliques
into $p$-cliques.

Perhaps surprisingly, it turns out that in many cases it is not worthwhile to ``start the search from scratch'': instead of using quantum search to detect $p$-cliques ``from scratch'', it is more efficient to first \emph{classically} list all $q$-cliques for some $q < p$, and then use quantum search to find an extension of some $q$-clique into a $p$-clique.
This echoes the theme of re-using information throughout the stages of the search: we precompute some information classically, which will be used by all stages of the search.
For example, we show that to solve triangle-detection, we can improve on the algorithm from~\cite{ILM20} by first classically listing \emph{edges}, so that every node of the congested clique learns some set of edges that it will be responsible for trying to complete into a triangle, and then using quantum search to find a node that forms a triangle with some edge. This reduces the running time from $\tilde{O}(n^{1/4})$ rounds in~\cite{ILM20} to $\tilde{O}(n^{1/5})$ rounds in our new algorithm.

More generally, we can solve the $K_p$-detection problem by first classically listing all instances of $K_{p-1}$
in the graph, and then performing distributed Grover search over the nodes, to check if some $(p-1)$-clique can be extended into 
a $p$-clique. Since there are $n$ nodes to check, the Grover search will require $\sqrt{n}$ quantum queries, and each query can be checked in $O(1)$ rounds (in the congested clique, it is possible to learn all neighbors of a given node in a single round).
Thus, the overall running time we obtain will be $\tilde{O}(L_{p-1} + \sqrt{n})$, where $L_{p-1}$ is the time required to list all $(p-1)$-cliques.
For example, $4$-cliques can be listed in $L_4 = O(\sqrt{n})$ rounds~\cite{DLP12}, and this approach allows us to solve the $K_5$-detection problem in roughly the same time complexity, $\tilde{O}(\sqrt{n})$.
However, the cost $L_{p-1}$ grows with $p$, so sometimes it is better to start from a smaller clique, $K_q$ for $q < p-1$,
and extend by more than a single node. This leads to our general nested-search-based approach, which starts by listing all copies of $K_q$
for some $q < p$, and then uses nested quantum search to check if some $q$-clique can be extended into a $p$-clique.

 \vspace{2mm}

\noindent\textbf{Review of prior works on classical algorithms for clique detection.}
In the \congest model, the first sublinear algorithm for $p$-clique detection was obtained for $p=3$ (i.e., triangle detection) by Izumi and Le Gall~\cite{ILG2017}. The complexity of triangle detection was then improved to $\tilde{O}(\sqrt{n})$ by Chang et al.~\cite{CPZ19}, where $n$ denotes the number of nodes, and then further to $\tilde{O}(n^{1/3})$ by Chang and Saranurak~\cite{CS19}. For $p$-cliques with $p\geq 4$, the first sublinear detection algorithm was constructed by Eden et al.~\cite{EFFKO19}. 
These results were improved to $\tilde{O}(n^{p/(p+2)})$ rounds for all $p\geq 4$ by Censor-Hillel et al.~\cite{CHLGL20}, and very recently, $\tilde{O}(n^{1-2/p})$ rounds for all $p\geq 4$ by Censor-Hillel et al.~\cite{Censor+SODA21}. Czumaj and Konrad \cite{CK20} have shown the lower bound $\Omega(\sqrt{n})$ for $p$-clique detection for $p\ge 4$, which matches the upper bound from \cite{Censor+SODA21} for $p=4$. Proving lower bounds for triangle detection, on the other hand, appears extremely  challenging: it is known that for any $\varepsilon > 0$, showing a lower bound of $\Omega(n^{\varepsilon})$ on triangle detection implies strong circuit complexity lower bounds~\cite{CHFGLGDR20} (see nevertheless \cite{ACHKL20} for a weaker, but still non-trivial, lower bound for triangle detection).\footnote{Note that the algorithms from \cite{CPZ19,CHLGL20,CS19,Censor+SODA21} actually solve the listing version of the problem (which asks to list all $p$-cliques of the graph) as well. For the listing version, lower bounds matching the upper bounds from \cite{CS19,Censor+SODA21} for all values of $p\ge 3$ are known \cite{ILG2017,PS18,FGKO18}.}  


In the powerful \clique model, the best known upper bounds on the round complexity of $p$-clique detection is $O(n^{0.158})$ for $p=3$, which is obtained by the algebraic approach based on matrix multiplication developed by Censor-Hillel et al.~\cite{CHKKLPJ19}, and $O(n^{1-2/p})$ for any constant $p\ge 4$ \cite{DLP12}. \vspace{2mm}


\noindent{\bf Further related works on quantum distributed computing.}
There exist a few works investigating the power of quantum distributed computing in other models  or settings (see also \cite{Arfaoui+14} and \cite{Denchev+08} for surveys). 
In the \local model, separations between the computational powers of the classical and quantum algorithms have been also obtained \cite{GKM09,LNR19}. Over anonymous networks, zero-error quantum algorithms have been constructed for leader election \cite{Tani+12}. Finally, quantum algorithms for byzantine agreements have also been investigated \cite{Ben-Or+STOC05}.\vspace{2mm}

\noindent{\bf Organization of the paper.} The core conceptual message of the paper is contained in the first 10 pages: we describe our main technique in Section \ref{section:prelim} and then, in Section \ref{sec:CC}, explain how  to use this technique to construct fast quantum algorithms for clique detection in the $\qclique$ model. Further sections then show how to apply the technique to construct fast algorithms in the $\qcongest$ model (Section \ref{sec:qcongest}, which proves Theorem \ref {th:clique-congest}) and construct even faster algorithms in the \textsc{Quantum}  \textsc{Congested Clique} model (Section \ref{sec:CCfast}, which proves Theorem \ref{th:clique-clique}). A proof of Theorem \ref{thm:Kp_obstacle} is given in Section \ref{sec:barrier}.

\section{Nested Distributed Quantum  Searches}
\label{section:prelim}
In this section we present our main technique: nested distributed quantum searches.
This is a generalization of a technique (called below \emph{distributed Grover search}) used in prior quantum distributed works \cite{Izumi+PODC19, LeGall+PODC18,ILM20}.

We note that implementing distributed quantum searches in a nested way is already allowed (but not used) in the framework introduced in \cite{LeGall+PODC18}. Our main contribution in the current section is developing this approach into a full framework and describing its concrete implementation in the distributed setting.

\paragraph{Standard Grover search.}
We begin by informally describing the most standard framework for Grover search ---
as a technique to solve a search problem with black-box access. 

Consider the following: given black-box access to a function $f\colon X\to\{0,1\}$, for an arbitrary $X$, find an $x\in X$ such that $f(x)=1$, if such an element exists. Grover's quantum algorithm \cite{Grover96} solves this problem with high probability using ${O}(\sqrt{|X|})$ calls to the black box. Grover's algorithm consists of ${O}(\sqrt{|X|})$ steps,
where each step executes one quantum operation called the \emph{Grover diffusion operation}, which does not use the black-box, and an operation called the \emph{checking procedure}, which uses one call to the black-box.

\paragraph{Distributed Grover search.} Let us present the basic quantum distributed search framework (distributed Grover search) introduced in \cite{LeGall+PODC18}. In this distributed implementation, one specific node,
called \emph{the leader}, run each Grover diffusion locally,
but the checking procedure is implemented via a distributed algorithm.

Consider again a function $f\colon X\to\{0,1\}$, for an arbitrary $X$, and the following search problem: one specified node (the leader) should find an element $x\in X$ such that $f(x)=1$, or, if no such element exists the leader should output ``not found''.  Assume there exists a distributed algorithm $\mathcal{A}$, called the \emph{checking procedure}, in which the leader is given $x\in X$ as input, and the leader returns $f(x)$ as output. The checking procedure~$\mathcal{A}$ is often described as a classical algorithm, but it can also be a quantum distributed algorithm.\footnote{As explained in \cite{LeGall+PODC18}, a classical procedure can easily be converted using standard techniques into a quantum procedure able to deal with superpositions of inputs.}

Let $r$ be the round complexity of~$\mathcal{A}$. The framework introduced in \cite{LeGall+PODC18} shows that there is a quantum distributed algorithm that runs in $\tilde{O}(\sqrt{|X|}\cdot r)$ rounds and enables the leader to solve the above search problem with probability at least $1-1/\poly(n)$. While the original statement in \cite{LeGall+PODC18} was for the \qcongest model, as explained in \cite{Izumi+PODC19}, the same holds for the \qclique model.

\begin{lemma}[\cite{LeGall+PODC18}]
\label{lemma:qsearchv0}
There is a quantum algorithm that runs in
$\tilde{O}(\sqrt{|X|}\cdot r)$
rounds and enables the leader to solve the above search problem with probability at least $1-1/\poly(n)$. This statement holds  in both the \qcongest and the \qclique models.
\end{lemma}

\noindent{\bf Nested distributed quantum searches.} All prior works using quantum distributed search (\cite{Izumi+PODC19, LeGall+PODC18,ILM20}) used a classical checking procedure $\mathcal{A}$. The framework of \cite{LeGall+PODC18} nevertheless allows quantum checking procedures. In particular, a distributed Grover search can be used as the checking procedure.
We now present our framework for nested distributed quantum searches, consisting of $k$ nested levels, where at each level:
\begin{enumerate}
	\item The nodes run a distributed \emph{setup} step for the current level, collecting information
		and preparing for the next search levels. The setup procedure in our results is classical, but in general it can be quantum.
	\item We execute the next level of the search. Crucially, the information prepared during the setup 
		will be re-used to evaluate all the nested queries in the next level (and subsequent levels).
\end{enumerate}

Formally, let $f\colon X_1\times\cdots\times X_k\rightarrow \{0,1\}$ be a function, for a constant $k\ge 2$ and sets $X_1,\ldots,X_k$. The goal is finding $(x_1,\ldots,x_k)\in X_1\times\cdots\times X_k$ where $f(x_1,\ldots,x_k)=1$, if such exists.
For $\ell\in\{1,\ldots,k-1\}, u\in V$, let
$\setup^u_\ell\colon X_1\times \ldots \times X_{\ell}\to\{0,1\}^\ast$ be a function describing the setup data of $u$ for the $(\ell+1)$-th search.
Let $\mathcal{S}_1,\ldots,\mathcal{S}_{k-1}$ and $\mathcal{C}$ be distributed algorithms with the following specifications.
\begin{itemize}
\item
\underline{Algorithm $\mathcal{S}_1$}. Input: the leader is given $x_1\in X_1$.  Output: each node $u\in V$ outputs $\setup^u_1(x_1)$.
\item
\underline{Algorithm $\mathcal{S}_\ell$} for any  $\ell\in\{2,\ldots,k\}$. Input: the leader is given $(x_1,\ldots,x_{\ell})\in X_1\times\cdots\times X_{\ell}$ and each node $u\in V$ is given $\setup^u_{\ell-1}(x_1,\ldots,x_{\ell-1})$. Output: each node $u\in V$ outputs $\setup^u_{\ell}(x_1,\ldots,x_{\ell})$.
\item
\underline{Algorithm $\mathcal{C}$}. Input: the leader is given $(x_1,\ldots,x_k)\in X_1\times\cdots\times X_k$ and each node $u\in V$ is given $\setup^u_{k-1}(x_1,\ldots,x_{k-1})$. Output: the leader outputs $f(x_1,\ldots,x_k)$.
\end{itemize}
Let $s_1,\ldots,s_{k-1}$ and $c$ denote the round complexities of $\mathcal{S}_1,\ldots,\mathcal{S}_{k-1}$ and $\mathcal{C}$, respectively. Applying Lemma~\ref{lemma:qsearchv0} leads to the following result.

\begin{lemma}
\label{lemma:qsearch}
There is a quantum algorithm that runs in
\[\tilde{O}\left(\sqrt{|X_1|}\left(s_1+\sqrt{|X_2|} \left(s_2+\sqrt{|X_3|}\left(s_3+\ldots+\sqrt{|X_{k-1}|}\left(s_{k-1}+\sqrt{|X_k|} \left( s_k + c\right) \right)\right)\right)\right)\right)\]
rounds and enables the leader to output $x_1,\ldots,x_k$ such that $f(x_1,\ldots,x_k)=1$, or output that there are no such $x_1,\ldots,x_k$, with probability at least $1-1/\poly(n)$. This statement holds  in both the \qcongest and the \qclique models.
\end{lemma}

\section{The Power of Nested Quantum Search: Clique-Detection from Listing in the $\qclique$ Model}
\label{sec:CC}
In this section we describe our approach
for taking an algorithm for $K_p$-listing in the \clique or \qclique model,
and extending it to $K_{p+t}$-detection (for some $t > 0$) using quantum search.
We give two variations of the approach:
the first uses the $K_p$-listing algorithm \emph{as a black box}, so that any such algorithm can be used
(for example, algorithms that perform better on certain classes of input graphs, etc.).
This also
forms the basis of our $\qcongest$ algorithms in Section \ref{sec:qcongest}.
The second approach yields faster results, but it
``opens the black box'', relying on the properties of the $K_p$-listing algorithm from \cite{DLP12}.
Since it is more complicated, the second approach is described in Section~\ref{sec:CCfast}.


To exploit the large bandwidth of the \clique, we use Lenzen's routing scheme
for solving the \emph{information distribution task}:
for some $s \geq 1$, each $v \in V$ has at most $s \cdot n$ messages $m^v_1,\dots,m^v_r$, each of $O(\log n)$ bits, and each with a destination $\mathit{dest}(m^v_i) \in V$.
Each $v \in V$ is the destination of at most $s \cdot n$  messages ($|\{m^u_i \mid u \in V \land i \in [s \cdot n] \land \mathit{dest}(m^u_i) = v\}| \leq s \cdot n$),
and we wish to deliver each message $m^v_i$ to its destination $\mathit{dest}(m^v_i)$.

\begin{lemma}[Lenzen's Routing Scheme\cite{Lenzen13}]
	\label{lem:lenzen_routing}
	The information distribution task with parameter $s$ can be solved in $O(s)$ rounds in \clique.
\end{lemma}


\subsection[Warmup: Detecting triangles in $\tilde{O}(n^{1/5})$ rounds]{Warmup: Detecting Triangles in {\boldmath $\tilde{O}(n^{1/5})$} Rounds}\label{sec:tri}
We describe a simple triangle detection algorithm demonstrating the basic idea of our approach, and improving
upon the state-of-the-art algorithm from~\cite{ILM20}.

In the algorithm,
we partition the search-space $V^3$ into $n$ \emph{shards}, one per node, and each $v \in V$
checks if there is a triplet $(u_1, u_2, u_3)$ in its shard that is a triangle in $G$.
Each shard has the form $A_i \times A_j \times Q_k$,
where $A_1,\ldots,A_{n^{2/5}}$ partitions $V$
into sets of $n^{3/5}$ nodes,
and $Q_1,\ldots,Q_{n^{1/5}}$ partitions $V$ into sets of size $n^{4/5}$.%
\footnote{To simplify the presentation,
here and everywhere in the paper,
when partitioning $V$ into $n^{\delta}$ subsets, for a $\delta \in (0,1)$,
we assume $n^{\delta}$ is an integer and divides $n$.
If this is not the case, one can replace $n^{\delta}$ by $\lceil n^{\delta} \rceil$, without affecting the asymptotic complexity.}%
Note that the total number of shards is indeed $n^{2/5} \cdot n^{2/5} \cdot n^{1/5} = n$.

To check if its shard $A_i \times A_j \times Q_k$ contains a triangle,
node $v$ learns the edges $E(A_i, A_j) = E \cap \left( A_i \times A_j \right)$, and then,
using a distributed quantum search, checks if some $w \in Q_k$ forms a triangle
with some $\set{ u_1, u_2} \in E(A_i, A_j)$.
The search is not performed directly over $Q_k$: instead, we partition $Q_k$ into \emph{batches},
$Q_k^1,\ldots,Q_k^b$,
and search for a batch $Q_k^{\ell}$ with a node forming a triangle.
 Processing the nodes in batches 
allows us to 
fully utilizes the bandwidth in the \clique.
However, we must balance
the \emph{size} of batches, which determines the time to check if a batch has a node completing
a triangle, against the \emph{number} of batches, which determines the number of quantum queries we will need to perform.

\paragraph{Detailed description of the algorithm.}

\noindent Consider a node $v$, and let $A_i \times A_j \times Q_k$ be the shard assigned to node $v$.
The algorithm has two steps:
\begin{enumerate}
	\item Node $v$ learns $E( A_i, A_j )$, using Lenzen's routing scheme.
		
	\item Node $v$ partitions $Q_k$ into $n^{2/5}$ batches,
		$Q_k = \{ Q_k^1,\ldots, Q_k^{n^{2/5}}\}$,
		each containing $n^{2/5}$ nodes (since $|Q_k| = n^{4/5}$).
		We use a quantum search over $\ell \in [n^{2/5}]$
		to check whether there exists a $Q_k^{\ell}$
		containing a node $w \in Q_k^{\ell}$ that forms a triangle together with two nodes $u_1 \in A_i, u_2 \in A_j$.
\end{enumerate}
Formally, we instantiate Lemma~\ref{lemma:qsearchv0}
with the search-space $X = [n^{2/5}]$ (i.e., the batch indices).
The checking procedure $\mathcal{A}$
checks an index $\ell \in [n^{2/5}]$
by routing $E(A_i \cup A_j, Q_k^{\ell})$
to $v$ (in parallel at all nodes).
Then, $v$ locally checks whether there is a $(u_1, u_2, w) \in A_i \times A_j \times Q_k^{\ell}$
such that $\set{ u_1, u_2 } \in E(A_i, A_j), \set{ u_1, w} \in E(A_i, Q_k^{\ell})$,
and $\set{ u_2, w} \in E(A_j, Q_k^{\ell})$;
it sends '1' to the leader if it found such a triplet, and '0' otherwise.

\paragraph{Complexity.} Step 1 requires $O( |A_i| \cdot |A_j| / n) = O(n^{2 \cdot 3/5 - 1}) = O(n^{1/5})$ rounds, using Lenzen's routing scheme (Lemma~\ref{lem:lenzen_routing}).
In Step 2, checking a particular batch $Q_k^{\ell}$ requires node $v$ to learn $E( A_i \cup A_j, Q_k^{\ell} )$.
As $|E( A_i \cup A_j, Q_k^{\ell} )| = O(n^{3/5 + 2/5}) = O(n)$, this can be done in $O(1)$ rounds using Lemma~\ref{lem:lenzen_routing}.
By Lemma~\ref{lemma:qsearchv0}, since the search space is $[n^{2/5}]$,
Step 2 takes $\tilde{O}(n^{1/5})$ rounds.
In total, the algorithm takes $\tilde{O}(n^{1/5})$ rounds.

\subsection[Extending $K_p$ \emph{listing} to $K_{p+t}$ \emph{detection}]{Extending {\boldmath $K_p$} \emph{Listing} to {\boldmath $K_{p+t}$} \emph{Detection}}

Our triangle detection algorithm 
has the following structure:
we view a triangle as an edge $\set{ u_1, u_2 }$, plus a node $w$
connected to $u_1, u_2$.
We \emph{classically} route information between nodes, 
so they can list the edges $\set{ u_1, u_2}$ in the sets they are responsible for ($A_i \times A_j$).
Then, we use \emph{quantum search} to check if there is a node $w$ forming a triangle with a listed edge.

We extend this idea to cliques of arbitrary sizes:
given $q > 2$, take $p,t$ where $p+t = q$.
We view a $q$-clique as a $p$-clique $\set{ v_1,\ldots,v_p}$, plus a $t$-clique $\set{ u_1,\ldots,u_t}$
where $u_1,\ldots,u_t$ are all connected to $v_1,\ldots,v_p$.
We classically 
list all $p$-cliques,
and then use quantum search to check for a $t$-clique forming a $q=(p+t)$-clique with a listed $p$-clique.

We present two variants of this approach. The first takes a $K_p$-listing algorithm as a black box,
making no assumptions about which $p$-cliques are found by which nodes.
The second improves on the first by ``opening the black box'' and using properties of the $K_p$-listing
algorithm of~\cite{DLP12}: knowing which $p$-cliques are listed by each node reduces the amount of information we route during the quantum search, as some edges are not relevant to some nodes.
We present the first variant here, and the second is given in Section~\ref{sec:CCfast}.
Note that our triangle detection algorithm is an instance of the second variant, since we exploit out knowledge of $A_i, A_j$ to determine which edges ($E(A_i \cup A_j, Q_k^{\ell})$) are learned by a given node as it evaluates batch $Q_k^{\ell}$.

How should we explore the search-space $V^{(t)}$ of possible $t$-cliques that may extend a given $p$-clique to a $(p+t)$-clique?
One possibility is to partition it into batches, and search over them, as we did for triangles.
However, the large search-space makes this inefficient: every node must learn
the edges between every pair of nodes in the current batch, and 
since we can use at most $n^2$ batches%
\footnote{Otherwise we will need more than $\sqrt{n^2} = n$ quantum queries.}
to cover $V^{(t)}$,
very soon we reach a situation where every node needs to learn \emph{all} the edges.
Instead, we use a \emph{nested search}, 
building the $t$-clique node-by-node.
The search is structured so that edges learned at a given level are re-used to evaluate \emph{many}
nested queries on following levels.
See Fig.~\ref{fig:search} for an example partitioning of the search-space.

\begin{figure}[h]
	\begin{center}
		\includegraphics[width=0.8\textwidth]{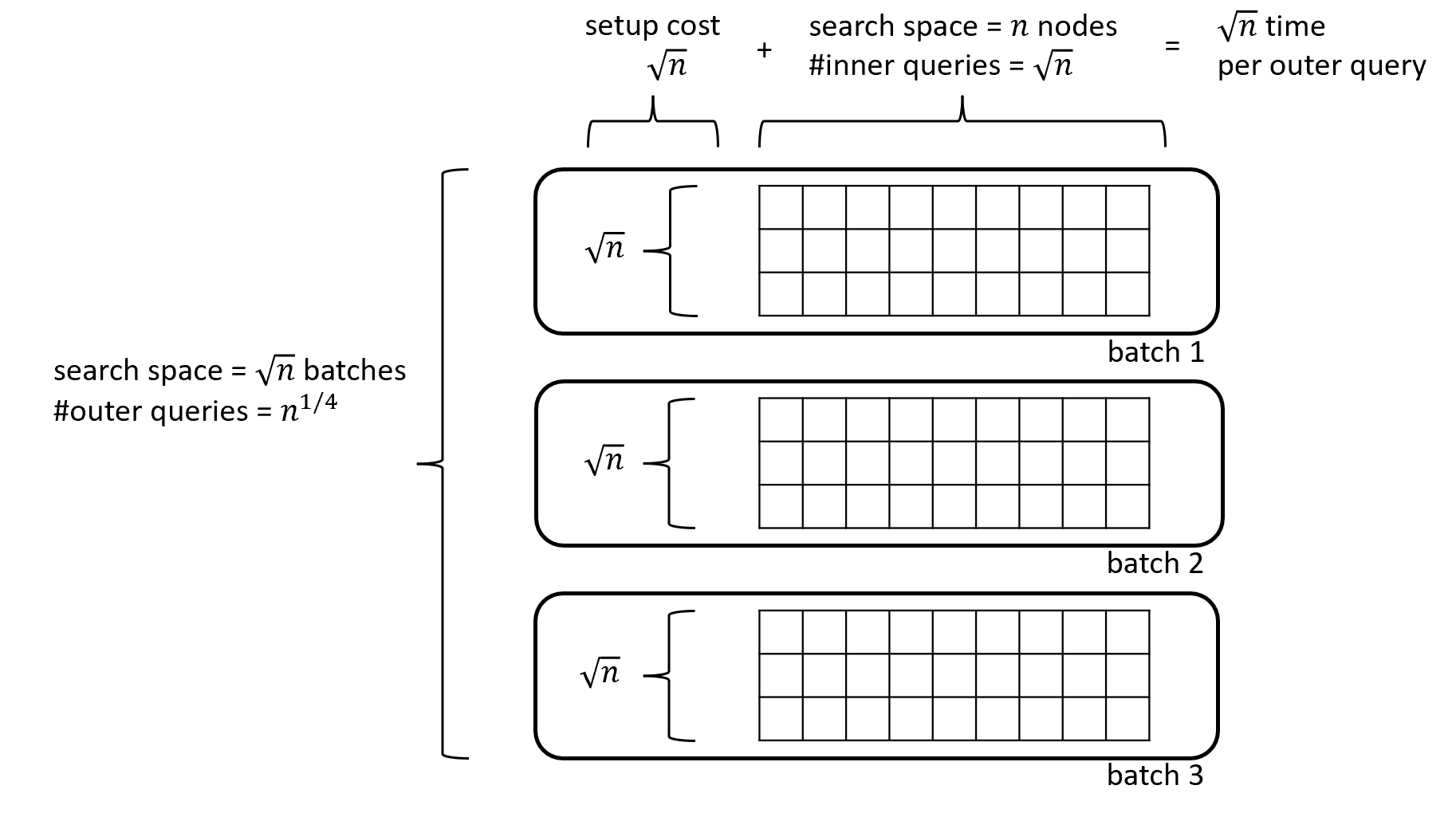}
	\end{center}
	\caption{
	Extending from $K_p$-listing to $K_{p+2}$-detection.
	We search $V \times V$ to check for a pair of nodes that can be added to an existing $p$-clique to form a $(p+2)$-clique (in the figure, $|V| = 9$).
	A non-nested search requires $\sqrt{n \cdot n} = n$ queries. In contrast, in a nested search, we split the first dimension into $\sqrt{n}$ batches, and search over them (requiring $n^{1/4}$ outer queries); to implement each outer query, all nodes send the edges corresponding to the current batch (requiring $\sqrt{n}$ rounds), and then explore the second dimension ($|V| = n$) using $\sqrt{n}$ inner quantum queries.
	The time of the entire search is $\tilde{O}\left( n^{1/4} \left( \sqrt{n} + \sqrt{n}  \right) \right) = \tilde{O}\left( n^{3/4} \right)$.
		}
	\label{fig:search}
\end{figure}

\paragraph{The initial state.}
Let $V^{(p)}$ denote all subsets of $V$ with cardinality $p$.
When we begin, we assume copies of $K_p$ have already been found:
each node $u \in V$ has a subset $S_u^p \subseteq V^{(p)}$
of $p$-cliques it found.
Let
	$S^p = \bigcup_{u \in V} S_u^p$
be all copies of $K_p$ found by the nodes.
We assume $S^p$ is the set of all $p$-cliques in $G$.

We say that an algorithm $A$ \emph{extends} from $K_p$ to $K_{p+t}$ if,
given sets $\set{ S_u^p }_{u \in V}$, 
w.h.p.,
algorithm $A$ outputs '1' at all nodes
iff $G$ contains a $(p+t)$-clique $\set{ v_1,\ldots,v_{p+t} } \in V^{(p+t)}$ such that
$(v_1,\ldots,v_p) \in S^p$.

\newcommand{\ThmExtending}
{
	For every $p \geq 2$, $t \geq 1$,
	there is an algorithm that extends from $K_p$ to 
	$K_{p+t}$ in $\tilde{O}(n^{1-1/2^t})$ rounds in $\qclique$.
}

\begin{theorem}
	\label{extendingTheorem}
	\ThmExtending
\end{theorem}

\begin{proof}

	Fix in advance $t$ partitions of $V$,
	where the $\ell$-th partition
	divides $V$ into $n^{1/2^{t-\ell}}$ sets,
	$V_1^{\ell},\ldots,V_{n^{1/2^{t-\ell}}}$,
	each of size $n^{1-1/2^{t-\ell}}$.
	Note that the partitions become increasingly finer,
	until
	at level $\ell = t$ we have $n$ sets comprising a single node each.
	Our goal is to execute a nested quantum search to check if there is a $t$-tuple of indices
	$(x_1,\ldots,x_t)$, such that there exist $v_1 \in V_{x_1}^1,\ldots,v_t \in V_{x_t}^t$
	forming a $(p+t)$-clique together with some previously-listed $p$-clique 
	$\set{ w_1,\ldots,w_p } \in S^p$.

	We instantiate Lemma~\ref{lemma:qsearch},
	executing a nested search with $t$ levels
	over the domain $X_1 \times \ldots \times X_t$,
	where $X_{\ell} = [n^{1/2^{t - \ell}}]$ for each $\ell = 1,\ldots,t$.
	We search for an element satisfying the function
	\begin{multline*}
		f( x_1,\ldots,x_t ) = 1
		\Leftrightarrow
		\\
		\exists v_1 \in V_{x_1}^1 \ldots \exists v_t \in V_{x_t}^t
		\exists \set{ w_1,\ldots,w_p } \in S^p :
		\text{$\set{ w_1,\ldots,w_p,v_1,\ldots,v_t}$ is a $(p+t)$-cliqe in $G$}.
	\end{multline*}

	At level $\ell \leq t$ of the search, the setup we prepare takes the form 
	of sets $\set{ S_u^{p+\ell} }_{u \in V}$ at the nodes,
	where $S_u^{p+\ell} \subseteq V^{(p+\ell)}$ is a set of $(p+\ell)$-cliques in $G$
	that node $u$ has learned about.
	When initiating the search (``$\ell = 0$''), we are given $\set{ S_u^p }_{u \in V}$.
	The algorithm $\mathcal{S}_\ell$ that prepares the setup for level $\ell \geq 1$ is as follows:
	\begin{itemize}
		\item Every node $u \in V$ broadcasts the subset of $V_{x_{\ell}}^{\ell}$
			that it is neighbors with:
			\begin{equation*}
				N_{u, x_{\ell}} = \set{ v \in V_{x_{\ell}}^{\ell} : \set{ u,v} \in E }.
			\end{equation*}
		\item Locally, each node $u \in V$ prepares $S_u^{p+\ell}$,
			by listing all the $(p+\ell)$-cliques that it can form by taking a $(p+\ell-1)$-clique
			from $S_u^{p+\ell-1}$ and appending to it a node from $V_{x_{\ell}}^{\ell}$:
			\begin{multline*}
				S_u^{p+\ell} = \set{ \set{ v_1,\ldots,v_{p + \ell} } \in V^{(p+\ell)} : 
					\set{ v_1,\ldots,v_{p+\ell - 1}} \in S_u^{p + \ell - 1},
				\right.
				\\
				\left.
					v_{p+\ell} \in V_{x_{\ell}}^{\ell},
					\text{ and 
					for each $i = 1,\ldots,p+\ell-1$ we have }
					v_{p+\ell} \in N_{v_i, x_{\ell}}
					}
				.
			\end{multline*}
	\end{itemize}
	The final algorithm $\mathcal{C}$ that evaluates $f( x_1,\ldots,x_{t} )$
	simply has each node $u$ inform the leader whether $S_u^{p+t}$ is empty or not.
	If there is some node $u$ with $S_u^{p+t} \neq \emptyset$,
	the leader outputs '1', and otherwise '0'.

	\paragraph{Complexity.}
	The size of the $\ell$-th level partition
	is chosen to as to
	balance the setup cost 
	against the time required for the remainder of the nested search:
	at level $\ell$, the setup cost is $|V_{x_{\ell}}^{\ell}| = s_{\ell} = n^{1-1/2^{t-\ell}}$.
	Using Lemma~\ref{lemma:qsearch},
	a backwards induction on $\ell$ shows the cost for levels $\ell+1,\ldots,t$
	of the search is $\tilde{O}\left(n^{1-1/2^{t-\ell}}\right)$,
	matching the setup cost.
	The cost of the entire search (i.e., levels $\ell = 1,\ldots,t$) is $\tilde{O}\left( n^{1-1/2^{t}} \right)$.
\end{proof}

By combining the classical clique-listing algorithm of~\cite{DLP12}
with Theorem~\ref{extendingTheorem}, we obtain an algorithm for detection of $K_p$ which improves
on the state-of-the-art classical algorithm for $p \geq 5$.

\begin{theorem}
	Given a graph $G = (V, E)$ and a clique size $p \geq 5$, it is possible to detect whether there exists an instance of $K_p$ in $G$ within $\tilde{O}(\min_{t \in \mathbb{N}} \max \{n^{1-\frac{2}{p-t}}, n^{1-\frac{1}{2^t}}\})$ rounds of the \qclique model.
\end{theorem}

For instance, for $p = 5$, by taking $t = 1$, we get $K_5$-detection in $\tilde O(n^{1/2})$ rounds of the \qclique model,
improving on the classical runtime of $\tilde{O}(n^{3/5})$.
We note that $\tilde{O}(n^{1/2})$ is the time required to classically list $K_4$, so the quantum-search-based
extension from 4-cliques to 5-cliques is ``for free''.

\section{Detection from Listing in the \qcongest Model}\label{sec:qcongest}
This section is devoted to the proof of Theorem~\ref{th:clique-congest}. 
\begin{theorem-repeat}{th:clique-congest}
\THcongest
\end{theorem-repeat}

We begin with the proof for triangle detection in \qcongest, and then proceed with $K_p$ detection for $p\geq 7$. 

At a very high level, our algorithms in this section use the framework of decomposing the graph into clusters of high conductance and working within each cluster in order to find the required subgraph, and then recursing over the edges remaining outside of clusters. To work within a cluster, throughout this section we will use the following expander decomposition and routing theorems.

\paragraph{Preliminaries.}
We begin by defining the notions of mixing time and conductance, which are used in the context of the expander decomposition. We note that we do not use these definitions directly, but rather, we use previously-proven lemmas that use these properties to obtain efficient routing or simulation procedures on such graphs.

The \emph{conductance} of a graph $G$ is $\Phi(G)=\min_{S \subseteq V} \frac{|E(S, V\setminus S)|}{\min(|E(S,V)|, |E(V, V\setminus S)|)|}$, is the worst-case ration between the number of edges crossing a cut in the graph, and the number of edges contained on either side of the cut. 

The conductance is related to the \emph{mixing time} of the graph, which, informally, is the number of steps required for a random walk starting from any vertex $u$ to become close to its stationary distribution, where the probability of being at any given vertex $v$ is roughly $\mathit{degree}(v) / 2m$. The mixing time of a graph $G$ is denoted $\tau_{\textrm{mix}}(G)$, and it is related to the conductance as follows: $\Theta\left(\frac{1}{\Phi(G)}\right) \leq \tau_{\textrm{mix}}(G) \leq \Theta\left(\frac{\log n}{\Phi^2(G)}\right)$ (see Corollary 2.3 in~\cite{JS89}).

\begin{lemma}[\cite{CS19,Censor+SODA21}]\label{lem:decomp_soda}
	Let $\delta > 0$ such that $m = n^{1+\delta}$. For any $\epsilon \in (0,1)$, and constant $\gamma \in (0,1)$, there is a constant $a_\gamma > 0$ dependent only on $\gamma$, such that a decomposition can be constructed in $\tilde{O}(n^\gamma)$ rounds, with high probability, in which the edges of the graph are partitioned into two sets, $E_m,E_r$, that satisfy the following conditions:
	\begin{enumerate}
		\item Each connected component (cluster) $C$ of $E_m$ has conductance $\Phi(C)\geq (\epsilon/\log{n})^{a_\gamma}$, and has average degree at least $ \epsilon n^\delta$.
		\item For any cluster $C$ and node $v \in V_C$, $\deg_{V_C}(v) \geq (\epsilon/\log{n})^{a_\gamma}\deg_{V \setminus V_C}(v)$.
		\item $E_r \leq \epsilon m$. 
	\end{enumerate}
\end{lemma}

\begin{theorem}[Expander routing~\cite{CS19,Ghaffari+PODC17}]\label{thm-expander-routing}
Suppose $\mix(G) = \polylog (n)$ and let $0\leq \gamma\leq 1$ be a constant.
There is an $O(n^{\gamma})$-round algorithm that pre-processes the graph such that for any subsequent routing task where each node $v$ is a source and a destination of at most $L \cdot \deg(v)$ messages of $O(\log n)$ bits, all messages can be delivered in $L \cdot \log^{\alpha} (n)$ rounds of the \congest model, with high probability, where $\alpha$ is a constant that depends only on $\gamma$.
\end{theorem}

\subsection{{\boldmath $\tilde{O}(n^{1/5})$}-Round Triangle Detection in the $\qcongest$ Model} \label{subsec:triangles_congest}
In this section we show an $\tilde{O}(n^{1/5})$-round triangle detection algorithm in \qcongest. That is, we prove the following theorem, which is part of Theorem~\ref{th:clique-congest}.

\begin{theorem}\label{th:tri-congest}
There exists an algorithm that solves triangle detection with success probability at least $1-1/\poly(n)$ in the \qcongest model with complexity $\tilde O(n^{1/5})$.
\end{theorem}

As observed in \cite{CPZ19,CS19} (see also Theorem 4 in \cite{ILM20}), it is sufficient to solve this problem in high conductance graphs 
(with some additional edges incident to this graph) in order to obtain an algorithm for general graphs. Our \qcongest algorithm shares many similarities with the $\tilde{O}(n^{1/5})$ triangle detection quantum algorithm in the \qclique shown in Section~\ref{sec:CC}, but requires a more ``sparsity aware'' approach for it to work in this more restricted model. 


Specifically, in the $\Tri$ problem, as defined in \cite{ILM20}, the input network $G'=(V',E\inext \cup E\outext)$ is a connected network such that the mixing time of the graph 
$G\inext=(V(E\inext),E\inext)$ is at most $O(\polylog(n))$ (e.g., its conductance is at least $\Omega(1/\polylog(n))$), and each edge in $E\outext$ is incident to a node in $V(E\inext)$ so that $deg_{G\inext}(v) \geq deg_{G'-G\inext}(v)$ for every node $v \in G\inext$, and the goal of the network is to determine whether $G'$ is triangle-free.

\begin{theorem}[\cite{CPZ19,CS19}, Theorem 4 in \cite{ILM20}]
\label{th:tri-in-conductance}
	Assume that there exists an $r$-round distributed \qcongest algorithm $A$ that solves the $\Tri$ problem with probability at least $1-1/n^3$ and uses only the edges in $E\inext \cup E\outext$ for communication. Then there exists an $O(r\log{n}+n^{0.1})$-round \qcongest algorithm that solves the triangle finding problem over the whole graph $G=(V,E)$ with probability at least $1-1/\poly(n)$.
\end{theorem}

In the remainder of this subsection we describe an algorithm in \qcongest for the problem $\Tri$, which assumes the input network has the restrictions mentioned above. We then plug this algorithm into Theorem~\ref{th:tri-in-conductance}  and obtain Theorem~\ref{th:tri-congest} as claimed. Specifically, we prove the following.

\begin{theorem}
\label{th:triAlg}
There exists an algorithm that solves the $\Tri$ problem with success probability at least $1-1/n^3$ in the \qcongest model with complexity $\tilde O(n^{1/5})$.
\end{theorem}

Denote by $m\inext = |E\inext|$ and by $n\inext = |V(E\inext)|$ the number of edges and vertices inside the high conductance graph respectively, and by $\bar{m} = |E\inext \cup E\outext|, \bar{n} = |V'|$  the number of edges and vertices in $G'$ respectively. Denote $\delta > 0$ such that $\bar{m} = \bar{n}^{1+\delta}$, that is, $O(\bar{n}^{\delta})$ is the average degree in $G'$. \vspace{2mm}

\noindent\textbf{Computation units.} Since $deg_{G\inext}(v) \geq deg_{G'-G\inext}(v)$ for every node $v \in G\inext$, there is a constant $c > 1$ such that $\bar{m} \leq cm\inext$. 
We define a \emph{computation unit} to be a paired set $(v,\tilde{E})$ where $v \in V(E\inext), \tilde{E} \subseteq E\inext$, and such that $|\tilde{E}| = \lfloor \bar{n}^\delta/(2c) \rfloor$ and all edges in $\tilde{E}$ are incident to $v$. We say that the node $v$ is the \emph{core} of the computation unit. Notice that a node can be a core of multiple computation units. A set of computation units is called \emph{pairwise-disjoint} if for any two computation units $(v_1,E_1),(v_2,E_2)$ in the set, either the cores are different, i.e., $v_1 \neq v_2$, or the edges $E_1$ and $E_2$ are disjoint.

The reason we define computation units is as follows. In a nutshell, our algorithm will split the edges of the input into sets, such that each set defines a triangle finding \emph{task} of checking whether any three edges in the set form a triangle and such that all possible triplets of edges are checked. To do so, we need to assign all such tasks to the nodes. However, a node with a smaller degree can receive less information to check compared with a node with a higher degree. This is where the computation units come into play: Each core of a computation unit will use the edges of its unit to send and receive information, so that it can process edges associated with a single such task. Now, the bandwidth of a node with a high degree will be exploited by having this node be a core of more computation units, and hence it will be responsible for solving more such tasks, compared with a node of a smaller degree.

First, we show that we have sufficient computation units. Consider a node $v$ with internal degree at least $\bar{n}^\delta/2c$ and split its edges into disjoint sets of size $\lfloor \bar{n}^\delta/2c \rfloor$, and a remainder set of edges of size at most $\lfloor \bar{n}^\delta/2c \rfloor$. Mark each such set of size $\lfloor \bar{n}^\delta/2c \rfloor$ as a computation unit with the node $v$ as its core.

\begin{lemma}\label{lem:num_of_computation_units}
	In the process above, the set of marked computation units is pairwise-independent and is of size at least $\bar{n}/2$.
\end{lemma}


\begin{proof}
	We note that the set is indeed pairwise independent, as for each node $v$ the edges of the computation units with node $v$ are edge-disjoint.
	
	Recall that $\bar{m} \leq cm\inext$ for the constant $c>1$. 
	Assume by contradiction that there are less than $\bar{n}/2$ computation units marked. Consider a graph $\tilde{G}\inext$ that is defined as a subgraph of $G\inext$ after removing all edges participating in any marked computation unit. By the assumption on the number of computation units and by the bound on the number of edges in a computation unit, there are at most $(\bar{n}/2)\cdot \lfloor \bar{n}^\delta/2c \rfloor = \bar{n}^{1+\delta}/4c$ edges removed. This implies that the number of edges in $\tilde{G}\inext$ is at least $m\inext - \bar{n}^{1+\delta}/4c \geq c\bar{m} - \bar{n}^{1+\delta}/4c = \bar{n}^{1+\delta}(c-1/4c) \geq \bar{n}^{1+\delta}/2c$, where the last inequality holds for $c\geq \sqrt{3}/4$, which we have since $c>1$.
	
	This implies that the average degree in $\tilde{G}\inext$ is at least $\bar{n}^{1+\delta}/2cn\inext \geq \bar{n}^{1+\delta}/2c\bar{n} \geq \bar{n}^\delta/2c$, which implies that there is at least one node $v$ with degree $\bar{n}^\delta/2c$ in $\tilde{G}\inext$, which contradicts the process above, as $v$ could have marked an additional computation unit.
\end{proof}

The above process of constructing computation units can easily be computed by each node separately. Thus, in order to prove Theorem~\ref{th:triAlg}, we need a way to make the computation units globally known, in the sense that each has a unique identifier that is known to all nodes. To this end, we will use the following technical claim.
\begin{claim}[\cite{CPZ19} variant of Lemma 4.1]
	\label{claim:2approx}
	Let $G = (V, E)$ be a graph with $\polylog (n)$ diameter. Assume every node has some integer value $f(v)$ which fits in a single $O(\log n)$ bit message. It is possible to give every node $v$ in the network a new unique identifier $i_v \in [n]$, such that there is a globally known function $\tilde{f}:[n]\rightarrow \mathbb{N}$ that any node $u$ can locally compute any of its values, and such that $\tilde{f}(i_v)$ is a $2$-approximation to $f(v)$ for all $v \in V$. This algorithm requires $\polylog (n)$ rounds of the \congest model.
\end{claim}

We will also need the ability to have a leader node that coordinates quantum searches. To implement a leader, the diameter of the network naturally gets into the round complexity. Luckily, in \cite{EFFKO19}, it was shown that for any $H$-freeness problem (i.e., $H$-detection), we may assume without loss of generality that the network has small diameter, stated as follows. 
\begin{lemma}[\cite{EFFKO19}]\label{thm:diameter_reduce}
	Consider an $H$-freeness problem, where $|H|=k$.
	Let $\mathcal{A}$ be a protocol that solves $\mathcal{P}$ in time $T(n,D)$ with error probability $\rho=o(\frac{1}{n\log{n}})$.
	There is an algorithm $\mathcal{A}'$ that solves $\mathcal{P}$ with round complexity $\tilde{O}(T(n,O(k \log{n}))+k\log^2{n})$ and error probability at most $c\rho n\log{n}+\frac{1}{\mathrm{poly}(n)}$, for some constant $c$.
\end{lemma}
Using the above \emph{diameter reduction} technique, we assume that the diameter of the network is $O(\polylog(n))$. Moreover, we assume that the network computes a leader node $v^*$ and a BFS tree rooted at $v^*$, which can be done by the network in $O(\polylog(n))$ rounds.

The last ingredient that we need is the notion of $k$-wise independent hash functions. For two integers $a,b$ we say that a function $f:[a]\rightarrow[b]$ is a \emph{$k$-wise independent hash function} if for any $k$ distinct elements $x_1,\dots,x_k \in [a]$ and any $k$ elements $y_1,\dots,y_k \in [b]$, it holds that $\Pr(f(x_1) = y_1 \land \dots \land f(x_k) = y_k) = 1/b^k$.

We will use the following tail bound for bounding the size of a given bin in a $k$-wise independent hash function.
\begin{lemma}[\cite{SJSS93},Theorem 5(II)(b)]
	\label{thm:tail_bound}
	If $\rv{X}$ is the sum of $k$-wise independent random variables,
	each of which is confined to the interval $[0,1]$ and has $\mu = E[\rv{X}]$,
	then for $\alpha \geq 1$ and assuming $k \leq \lceil \alpha \mu e^{-1/3} \rceil$,
	\begin{equation*}
	\Pr(|X-\mu| \geq \alpha \mu) \leq e^{-\lfloor k/2 \rfloor}.
	\end{equation*}
\end{lemma}

We are now ready to prove Theorem~\ref{th:triAlg}.

\begin{proof}[Proof of Theorem~\ref{th:triAlg}]
First, each node splits its edges into computation units as describe above. This is done locally without communication. By Lemma~\ref{lem:num_of_computation_units}, there are at least $\bar{n}/2$ computation units.

Recall that by Lemma~\ref{thm:diameter_reduce}, we can assume that the diameter of $G'$ is $O(\polylog{n})$. The nodes run the procedure of Claim~\ref{claim:2approx} to give to each node $v$ a new unique ID $i_v$ in $[\bar{n}]$, such that there is a globally known function $\tilde{f}:[n]\rightarrow \mathbb{R}$ that any node $u$ can locally compute any of its values, and such that for any $v \in V$, $\tilde{f}(i_v)$ is a $2$-approximation to the number of computation units of $v$. Using this knowledge, the network gives a unique ID in $\lfloor \bar{n}/4 \rfloor$ to a set of $\lfloor \bar{n}/4 \rfloor$ marked computation units, such that all nodes know for each such computation unit its node ID (this can be done locally since the function $\tilde{f}(i_v)$ is globally known). Note that we have $\lfloor \bar{n}/2 \rfloor$ marked computation units, but we can only promise global knowledge of a 2-approximation on their number. This stage completes in $O(\polylog{n})$ rounds. 

Following this, using standard computation of a maximum value, the nodes of the network find the node $v^*$ in $G\inext$ with the highest ID and mark it as the leader node. The nodes then construct a BFS tree of $G'$ from $v^*$ and, using the BFS tree, $v^*$ propagates $O(\polylog{n})$ unused random bits to the rest of the network. Since the diameter of the network is at most $O(\polylog{n})$, the above needs at most $O(\polylog{n})$ rounds.

Using the the propagated random bits, the nodes choose at random two hash functions, $h_A: [\bar{n}]\rightarrow [\bar{n}^{2/5}]$ and $h_S: [\bar{n}]\rightarrow [\bar{n}^{1/5}]$, uniformly and independently from a $\Theta(\log{n})$-wise independent hash function family. Let $A=\{A_1,\dots,A_{\bar{n}^{2/5}}\}$ be the sets such that $A_i = \{v \mid h_A(v) = i\}$, and let $S=\{S_1,\dots,S_{\bar{n}^{1/5}}\}$ be the sets such that $S_i = \{v \mid h_S(v) = i\}$. 

Now, let $P$ be a partition of the computation units of $G\inext$ into  $\bar{n}^{4/5}$ sets of size at least $\bar{n}^{1/5}/4$ each. We arbitrarily index the sets in $P$ by $P_{i,j}$ for $1\leq i,j\leq \bar{n}^{2/5}$, and associate each set $P_{i,j}$ in $P$ with a pair of sets $(A_i,A_j)$. 
In addition, for each $1\leq i,j\leq \bar{n}^{2/5}$ and each computation unit $C_{i,j,\ell}$ in $P_{i,j}$ (for $\ell$ in a range from 1 to a value that is at least $\bar{n}^{1/5}/4$), we  associate $C_{i,j,\ell}$ with at most $4$ sets in $S$ (this is possible as there are at least $\bar{n}^{1/5}/4$ computation units in each $P_{i,j}$ and $\bar{n}^{1/5}$ sets in $S$). We denote the sets in $S$ that are associated with $C_{i,j,\ell}$ by $S_{i,j,\ell,1}, \dots, S_{i,j,\ell,4}$ (the indices $i,j$ can be omitted by being using the same 4 sets depending only on $\ell$ and being the same for every $i,j$).

The algorithm proceeds in two phases: For the first phase, consider a set of computation units $P_{i,j}$ in $P$, and define $E_1(i,j) = E(A_i,A_j)$. In the first phase, for each $1\leq i,j\leq  \bar{n}^{2/5}$, each core $v$ of a computation unit in part $P_{i,j}$ of $P$ learns all edges in $E_1(i,j)$.  As $h_A$ is a pairwise-independent hash function, for an edge $\{u,w\}$, we have $\Pr(h_A(u) = i \land h_A(w) = j) = (1/\bar{n}^{2/5})^2 = 1/\bar{n}^{4/5}$. 
Therefore, the number of edges that a core $v$ is required to learn in expectation is $\E(|E_1(i,j)|) = \bar{m}/\bar{n}^{4/5} = \bar{n}^{\delta+1/5}$. 

By the tail bound of Lemma~\ref{thm:tail_bound}, since we use an $\Omega(\log{n})$-wise independent hash function family (this is a logarithm of the total number of nodes $n$), the number of these edges is at most $\tilde{O}(\bar{n}^{\delta+1/5})$ w.h.p. (in $n$). Therefore, the core of a computation unit may learn these edges in $\tilde{O}(\bar{n}^{1/5})$ rounds using the routing scheme of Theorem~\ref{thm-expander-routing} with a sufficiently small $\gamma$ (since we have a degree of $O(\bar{n}^\delta)$ in $G\inext$ for the core using edges of this computation unit alone).

We now have that for every $i,j$, there are at least $\bar{n}^{1/5}/4$ computation units (those in $P_{i,j}$), whose cores know all edges in $E(A_i,A_j)$. Each such core is associated with 4 sets $S_{i,j,k,1}, \dots, S_{i,j,k,4}$, and what we would like to do in the second phase is for each core to check all edges from one of its $S$ sets to nodes in $A_i,A_j$ and detect a triangle. To leverage the power of the distributed Grover search, we split the task of each core into subtasks on which we can apply Lemma~\ref{lemma:qsearchv0}.

Formally, for the second phase, we define for each $S_i$ an arbitrarily split of $S_i$ into $\bar{n}^{2/5}$ batches, each of size $\tilde{O}(\bar{n}^{2/5})$ nodes, and denote these batches by $S_i^{(1)},\dots,S_i^{(\bar{n}^{2/5})}$. We define the following protocols $\mathcal{A}_1,\dots,\mathcal{A}_{\bar{n}^{2/5}}$. 
The protocol $\mathcal{A}_r$ is defined as follows:  For each $k$ such that $1\leq k\leq \bar{n}^{1/5}$ and for each $1 \leq r \leq \bar{n}^{2/5}$, let $E_2(i,j,k,r) = \{(u,w) \mid u\in S_k^{(r)} \land w \in A_i\cup A_j\}$. In the protocol, each core $v$ of a computation unit $C_{i,j,\ell}$ in $P_{i,j}$ learns all edges of $E_2(i,j,k,r)$ for at most 4 values of $k$ which are assigned to it. Since a set $S_k^{(r)}$ has $\tilde{O}(\bar{n}^{2/5})$ nodes, the number of edges to be collected between nodes of $S_k^{(r)}$ and $A_i \cup A_j$ is in expectation $\E(|E_2(i,j,k,r)|) = \bar{n}^{1+\delta}\cdot (\bar{n}^{2/5}/\bar{n}) \cdot (\bar{n}^{3/5}/\bar{n}) = \tilde{O}(\bar{n}^\delta)$.
By the tail bound of Lemma~\ref{thm:tail_bound}, the number of these edges is at most $\tilde{O}(\bar{n}^{\delta})$ w.h.p. (in $n$), therefore the computation unit may learn these edges in $\tilde{O}(1)$ rounds using the routing scheme of Theorem~\ref{thm-expander-routing}. A computation unit in $P_{i,j}$ rejects if there is a triangle contained in $E_1(i,j) \cup E_2(i,j,k,r)$ for any of the 4 values of $k$ which are assigned to it. Using the BFS tree, the network determines whether there was a computation unit that rejected, rejects if so, and otherwise accepts (notice that this propagation would not be needed in a non-quantum algorithm, as it is sufficient that one node rejects, but here we need the leader $v^*$ to know this information).

The way these protocols are exectued is as follows. In the second phase, the leader $v^*$ performs the quantum procedure described in Lemma~\ref{lemma:qsearchv0} on $\mathcal{A}_1,\dots,\mathcal{A}_{\bar{n}^{2/5}}$. Therefore, by Lemma~\ref{lemma:qsearchv0}, the quantum protocol rejects if and only if the network rejects in at least one of these procedures, and terminates after  $\tilde{O}(\sqrt{\bar{n}^{2/5}}) = \tilde{O}(\bar{n}^{1/5}) = \tilde{O}(n^{1/5})$ rounds. 

\paragraph{Correctness.}
If the graph is triangle-free then the network accepts, as a core $v$ of a computation unit only rejects if it detects a triangle in the edges $E_1(i,j) \cup E_2^r(k)$ for some appropriate parameters $i,j,r,k$. Otherwise, let $\{v_1,v_2,v_3\}$ be vertices of a triangle in $G'$. Let $i,j \in [\bar{n}^{2/5}]$ such that $v_1 \in A_i, v_2 \in A_j$ and let $P_{i,j}$ be the set of computation units to which the pair $A_i,A_j$ is mapped. Let $k \in [\bar{n}^{1/5}]$ and $r \in [\bar{n}^{2/5}]$ such that $v_3 \in S^{(r)}_k$ and let $w \in P_{i,j}$ be the core of a computation unit in $P_{i,j}$ to which $S_k$ is mapped. We note that the triangle is contained in $E_1(i,j) \cup E_2(i,j,k,r)$, and therefore the network rejects in the protocol $\mathcal{A}_r$, and hence our algorithm rejects as well. 
\end{proof}

\subsection[Quantum detection of $K_{p}$ for $p\geq 7$]{Quantum Detection of {\boldmath $K_{p}$} for {\boldmath $p\geq 7$}}
\label{sec:congest:seven}


We now show our method for enhancing the $K_{p}$-listing algorithm of \cite{Censor+SODA21} with a quantum procedure that allows us to detect a $K_{p+1}$ instance. 
%
%
To ease the notation, we will be using $p+1$ for the size of the clique that we are detecting, and hence note that the following statement is shifted, i.e., holds for $p \geq 6$.
\newcommand{\ThmCongestPseven}
{
There exists a quantum algorithm that solves $(p+1)$-clique detection with success probability at least $1-1/\poly(n)$ in the \qcongest model with complexity $\tilde{O}(n^{1-2/p})$, for $p \geq 6$.
}

\begin{theorem}
\label{theorem:K6detectQcongest}
\ThmCongestPseven
\end{theorem}

 In a nutshell, the algorithm uses a conductance decomposition, and works on clusters in parallel.
 There are three ways in which the $K_p$-listing algorithm of \cite{Censor+SODA21}
 may find instances of~$K_p$. The first two ways involve having a node learn its induced 2-hop neighborhood, and in these cases,
 this clearly gives detection of a $K_{p+1}$ instance if such an instance exists, without further effort on our part.
 The third case is where we diverge from the algorithm of \cite{Censor+SODA21} by incorporating Grover searches
 inside the clusters, using Lemma~\ref{lemma:qsearchv0}.
 
During the algorithm, some parts use a simple listing of the edges in a node's induced $2$-hop neighborhood. For completeness, we state here a formal claim and proof of how this is done.
 \newcommand{\ClaimTwoHop}
{
	Given a graph $G = (V, E)$, and some value $\alpha$, every node $v$ such that $\deg(v) \leq \alpha$ can learn its induced $2$-hop neighborhood in at most $O(\alpha)$ rounds of the \congest model.

}
\begin{claim}
	\label{claim:2hopNeighborhood}
	\ClaimTwoHop
\end{claim}

\begin{proof}[Proof of Claim~\ref{claim:2hopNeighborhood}]
	Let $v$ be such a node. Node $v$ iterates over its at most $\alpha$ neighbors. When it considers neighbor $i$, it sends a message to all its neighbors asking if they neighbor $i$, to which they each respond whether or not they have an edge to node $i$. At the end of the iterations, after $O(\alpha)$ rounds, node $v$ knows all of its induced $2$-hop neighborhood.
	
	Notice that it does not matter if $v$ has a neighbor $u$ which also tries to perform this search in parallel to $v$, in the odd rounds $v$ and $u$ just sends queries to each other across their shared edge, and in the even rounds they each respond to one another.
\end{proof}

We are now ready to prove Theorem~\ref{theorem:K6detectQcongest}.
 
\begin{proof}[Proof of Theorem~\ref{theorem:K6detectQcongest}]
We begin by briefly explaining the $K_p$ listing algorithm of \cite{Censor+SODA21}. The algorithm of \cite{Censor+SODA21} works in iterations, which are composed of stages. In every stage, some $K_p$ in the graph may be listed.  We split the listed $K_p$ instances according to the stage in the algorithm in which they are listed. For each stage, we show how to detect in a quantum manner if there is an instance of $K_p$ that is listed in that stage which can be extended to (i.e., is a part of) an instance of $K_{p+1}$. As such, if there is any instance of $K_{p+1}$ in the graph, then we will certainly detect this as some instance of $K_p$ could be extended to it.

We now show the $K_p$ listing algorithm, and interject at appropriate places in order to perform quantum searches. Each iteration of the algorithm of \cite{Censor+SODA21} consists of the following steps.\vspace{2mm}
		\begin{itemize}
\item[{\bf 1.}]{\bf Exhaustive search detection stage.} Every node $v$ with degree $O(n^{1/2})$ learns its induced 2-hop neighborhood in $O(n^{1/2})$ rounds using Claim~\ref{claim:2hopNeighborhood}, lists any cliques which it sees involving itself, and removes itself, along with its incident edges from the graph.

\textbf{Detection of \boldmath{$K_{p+1}$}:} In particular, because the entire induced 2-hop neighborhood is learned, any instance of $K_{p}$ that is listed that can be extended to an instance of $K_{p+1}$ can be immediately listed in this way even without further communication, which proves our claim for this stage.
\item[{\bf 2.}]{\bf Graph decomposition and cluster exhaustive search detection stage.} An expander decomposition is computed according to Lemma~\ref{lem:decomp_soda}, with $\gamma=0.1$ (any small constant would do here), while the above ensures that the average degree in the graph (and thus in the clusters, due to Lemma~\ref{lem:decomp_soda}) is at least $\Omega(n^{1/2})$. In clusters with $O(n^{1-2/p})$ nodes, each node learns its induced 2-hop neighborhood in $G$ using Claim~\ref{claim:2hopNeighborhood} and lists all the instances of $K_p$ which it is a part of, and finally removes itself and its incident edges from the graph. Because Lemma~\ref{lem:decomp_soda} promises that the number of edges that leave $C$ is at most $\tilde{O}(|E_C|)$, where $E_C$ is the set of edges inside $C$, then this completes in $\tilde{O}(n^{1-2/p})$ rounds.

			\textbf{Detection of \boldmath{$K_{p+1}$}:}  As in the previous stage, because of the exhaustive search nature of this stage, we again obtain that any instance of $K_{p}$ that is listed that can be extended to an instance of $K_{p+1}$ can be immediately listed in this way even without further communication, which proves our claim for this stage.\vspace{2mm}
\item[{\bf 3.}]{\bf Cluster listing.} Within every remaining cluster $C$, some nodes are designated as \emph{good} and the rest as \emph{bad}. For a good node $v$, the nodes of $C$ learn amongst themselves (collectively, not necessarily by a single node, and in particular not necessarily by $v$ itself) all the edges in the induced 2-hop neighborhood of $v$. That is, the nodes of $C$ request from the nodes neighboring the cluster to send in edges from outside $C$ to nodes inside $C$, such that all the induced 2-hop neighborhood of $v$ is known to the nodes in $C$. Now, the nodes in $C$ ensure that every instance of $K_p$ involving $v$ and at least one other node in $C$ is becomes known to some node in $C$ -- that is, these $K_p$ instances are listed by $C$. The good nodes are then removed from the graph along with all of their edges. 
%
\end{itemize}

\noindent This concludes the description of the \congest algorithm in~\cite{Censor+SODA21}. The exact definition of good nodes is not required for our purpose.
The proof of \cite{Censor+SODA21} shows that each iteration completes within $\tilde{O}(n^{1-2/p})$ rounds, and that after $\poly \log (n)$ iterations the remaining graph is empty and hence $\tilde{O}(n^{1-2/p})$ rounds are sufficient for the entire algorithm.\vspace{3mm}

		\noindent{\bf The cluster quantum detection stage.}
		We now provide the quantum procedure that we run after the above \textbf{Cluster listing} stage, in order to detect an instance of $K_{p+1}$ that contains an instance of $K_p$ that was listed during this stage. First, let the $C$-degree of a node be its number of neighbors in $C$,
let $\mu_{C}$ be
		the average $C$-degree of nodes in $C$, and let $H(C)$ be the high $C$-degree cluster nodes, namely, those of $C$-degree $
\Omega(\mu_{C})$.  Notice that due to the invocation above of Lemma~\ref{lem:decomp_soda}, it holds that $\mu_C = \Omega(n^{1/2})$.
		A more precise description of the algorithm in~\cite{Censor+SODA21} is that the induced 2-hop neighborhood of each good node $v$ becomes known to the nodes in $H(C)$ rather than to any node in $C$. Second, denote by $N^{+}(C)$ the set of nodes which have a neighbor in $C$ (this includes all of $C$ as $C$ is connected, by definition) and by $N(C) = N^{+}(C)-C$ the neighbors of $C$ outside of $C$. In the above \textbf{Cluster listing} stage we piggyback the degree (not $C$-degree) of every node in $N(C)$ onto the messages containing its edges, so that the degrees of all nodes in $N(C)$ are also known to nodes of $H(C)$.
		
The quantum process starts with the following. The nodes in $C$ elect some arbitrary leader $v_C \in H(C)$.
This can be done in $\poly \log (n)$ rounds because this is a bound on the diameter of the cluster, by Lemma~\ref{lem:decomp_soda}. The nodes of $H(C)$ now broadcast within $H(C)$ the degrees of all the nodes in $N(C)$. This is done in two steps. In the first step, each node in $H(C)$ sends $v_C$ the degrees of nodes in $N(C)$ which it knows about.
As every node sends and receives at most $|N(C)| = O(n)$ pieces of information, and every node in $H(C)$ has $C$-degree at least $\Omega(\mu_C) = \Omega(n^{1/2})$, this takes at most $\tilde O(n^{1/2})$ rounds, using the routing algorithm of Theorem \ref{thm-expander-routing} with a sufficiently small $\gamma$. In the second step, $v_C$ makes this information known to all nodes in $H(C)$ using a simple doubling procedure: in each phase of this procedure, each informed node shares the information with a unique uninformed node. Each phase completes in $\tilde O(n^{1/2})$ rounds by the same argument using the routing algorithm of Theorem~\ref{thm-expander-routing}, and the number of phases is logarithmic. In a similar fashion, the nodes in $H(C)$ learn all the degrees of the nodes in $C$.

Knowing the degrees of all nodes in $N^{+}(C)$, the nodes $H(C)$ \emph{bucket} $N^{+}(C)$ by degrees. That is, they compute $N^{+}(C) = N^{+}_1(C), \dots, N^{+}_{\log n}(C)$, such that the degree of any $v \in N^{+}_i(C)$ is in $[2^{i-1},\dots,2^i)$. Our goal is to iterate over the $\log n$ buckets, whereby in each bucket we check whether a node $v$ in $N^{+}_i(C)$ can be used to extend some $K_p$ instance, which is already listed by $H(C)$, into an instance of $K_{p+1}$. This is done as follows.

Fix an $i$ between 1 and $\log n$. We perform a Grover search over $v \in N^{+}_i(C)$ using $v_C$ as the leader of the search: In each query (i.e., over such a $v$), we broadcast within $H(C)$ all the neighbors of $v$ which are known to $H(C)$, in order to try to extend any $K_p$ listed by $H(C)$ to a $K_{p+1}$ involving $v$. Let $\kappa_{p}$ be such a $K_p$. As it is listed by $H(C)$, it must involve at least one good node in $C$, denoted $g \in \kappa_{p}$. Further, recall that every edge in the induced 2-hop neighborhood of $g$ is known to $H(C)$. Thus, if $v$ that can be used to extend $\kappa_{p}$ to an instance of $K_{p+1}$, all the edges between $v$ and $\kappa_{p}$ are known to the nodes of $H(C)$. Therefore, if we manage to broadcast in $H(C)$ all the edges incident to $v$ that are known to $H(C)$, then if there is a way to use $v$ to extend a $K_p$ listed by $H(C)$ to a $K_{p+1}$, we will certainly find it. Similarly to the above analysis, using the routing algorithm of Theorem~\ref{thm-expander-routing} and the doubling procedure, we broadcast within $H(C)$ the edges incident to $v$ which are known to $H(C)$. This completes in $O(2^i / \mu_C+1)$ rounds, as the degree of $v$ is at most $2^i$.

In order to conclude the proof, we need to bound the size of each $N^{+}_i(C)$. iIn order to obtain the bound on the size of $N^{+}_i(C)$, we compute an upper bound on the number of \emph{all} edges incident to nodes in $N^{+}(C)$ and use the bound on degrees of nodes in the bucket. To this end, we wish to show that $\sum_{v \in N^{+}(C)} \deg(v) = \tilde O(n \cdot \mu_C)$. To do so, we split the edges into three categories: $E_1$ -- edges with both endpoints in $C$; $E_2$ -- edges with one endpoint in $C$; $E_3$ -- edges with both endpoints not in $C$. As $\mu_C$ is the average $C$-degree of the nodes in $C$, and $n$ is the number of nodes in the entire graph, implying $|C| \leq n$, it holds that $E_1 = O(n \cdot \mu_C)$. It is ensured in Lemma~\ref{lem:decomp_soda} that for every $v \in C$, it holds that $\deg(v) = \tilde O(\deg_C(v))$, where $\deg_C(v)$ is the $C$-degree of $v$, implying $|E_2| = \tilde O(|E_1|) = \tilde O(n \cdot \mu_C)$. Further, it is ensured in Lemma~\ref{lem:decomp_soda} that $|E_3|/n = \tilde O(|E_1|/|C|) = \tilde O(\mu_C)$, implying that $|E_3| = \tilde O(n \cdot \mu_C)$.
Thus, for every $N^{+}_i(C)$, due to the degrees in the bucket and due to the pigeonhole principle, it must be that $|N^{+}_i(C)| = \tilde{O}(n \cdot \mu_C / 2^i)$. Further, as $n$ is the number of nodes in the graph, it also trivially holds that $|N^{+}_i(C)| \leq n$. All in all, we get that $|N^{+}_i(C)| = \tilde{O}(\min \{n, n \cdot \mu_C / 2^i\})$.

Finally, we can analyze the round complexity of the algorithm. The $K_p$ listing algorithm is shown in \cite{Censor+SODA21} to take $\tilde O(n^{1-2/p})$ rounds. For each $i$, the checking procedure in our Grover search takes $\tilde O(2^i / \mu_C+1)$ rounds. For every $i$ such that $2^i/\mu_C \leq 1$, the checking procedure takes $O(1)$, and $|N^{+}_i(C)| = \tilde O(n)$, implying that the Grover search takes a total of $\tilde O(n^{1/2})$ rounds, using Lemma~\ref{lemma:qsearchv0}. For every $i$ such that $2^i/\mu_C > 1$, the checking procedure takes $\tilde O(2^i / \mu_C+1) = \tilde O(2^i / \mu_C)$, and $|N^{+}_i(C)| = \tilde{O}(n \cdot \mu_C / 2^i)$, implying that the Grover search takes a total of $\tilde O(\sqrt{n \cdot \mu_C / 2^i} \cdot (2^i / \mu_C)) = \tilde{O}(\sqrt{n \cdot (2^i / \mu_C)})$ rounds, using Lemma~\ref{lemma:qsearchv0}. As $2^i = O(n)$ and $\mu_C = \Omega(n^{1/2})$ (due to the guarantees of the invocation of Lemma~\ref{lem:decomp_soda}, as stated above), this takes at most $\tilde O(n^{3/4})$ rounds.

We thus have that our quantum algorithm for detecting an instance of $K_{p+1}$ for $p\geq 5$ completes in $\tilde{O}(n^{3/4}+n^{1-2/p})$ rounds, with probability at least $1-1/\poly(n)$.

We note that we can decrease the $\tilde{O}(n^{3/4})$ part of the complexity, by slightly changing the listing algorithm of \cite{Censor+SODA21}. That is, the algorithm stated there performs $K_p$ listing in $\tilde O(n^{1-2/p})$ rounds, while ensuring $\mu_C = \Omega(n^{1/2})$. It is implied in the proofs in \cite{Censor+SODA21} that, for any $1/2 \leq \delta < 1$, one can pay an additional $\tilde O(n^\delta)$ rounds (by performing an exhaustive search, as done in Claim~\ref{claim:2hopNeighborhood}) in order to ensure $\mu_C = \Omega(n^{\delta})$. Further, notice that our Grover searches take a total of $\tilde O(n^{1/2} + n/\mu_C^{1/2}) = \tilde O(n^{1/2} + n^{1-\delta/2})$ rounds.

All in all, for any $1/2 \leq \delta < 1$, our algorithm requires $\tilde O(n^{1-2/p} + n^\delta + n^{1/2} + n^{1-\delta/2}) = \tilde O(n^{1-2/p} + n^\delta + n^{1-\delta/2})$ rounds. One can set $\delta = 2/3$, giving a final algorithm running in $\tilde O(n^{1-2/p} + n^{2/3} + n^{1-1/3}) = \tilde O(n^{1-2/p} + n^{2/3}) = \tilde O(n^{1-2/p})$ rounds, where the last transition is since $p \geq 6$.
\end{proof}



%


\section{Faster Clique Detection in the $\qclique$ Model}
\label{sec:CCfast}
We present here an approach that improves upon our approach in Section~\ref{sec:CC} by 
taking into consideration \emph{which} clique nodes know of which copies of $K_p$:
instead of treating the $K_p$-listing algorithm as a black box, we explicitly use the $K_p$-listing algorithm of~\cite{DLP12},
so that we know which copies of $K_p$ will be listed by each clique node, and what other information that node already has.
We then search for $(p+t)$-cliques by having each clique node learn only the edges that it needs to check if the $p$-cliques it has listed can be extended.
(In the previous approach, edges could be learned by nodes that had no use for them, since they were not adjacent to any $p$-clique that the node had listed.)

Throughout this section we use the following notation: we let $\bigtimes(S_1,\ldots,S_k)$ denote the Cartesian product $S_1 \times \ldots \times S_k$.
Also, we denote by $E(S_1,\ldots,S_k)$ the edges $E \cap \bigcup_{i \neq j} \left( S_i \times S_j \right)$
that cross between any two sets $S_i, S_j$.


For presenting our algorithm for a general $p$, we first shortly review the listing algorithm of Dolev et al.~\cite{DLP12}.
Fix an arbitrary partition $S_1,\ldots,S_{n^{1/p}}$ of the nodes of $V$,
such that $|S_i| = n^{1-1/p}$ for each $i$,
and a 1:1 mapping 
$g : V \rightarrow \set{ 1,\ldots,n^{1/p} }^p$ 
assigning to each node $v \in V$ a $p$-tuple $g(v) = (i_1,\ldots,i_p) \in \set{1,\ldots,n^{1/p}}^p$.
We assume for simplicity that $n^{1/p}$ is an integer; otherwise, we can adjust the set size to an integer,
without changing the asymptotic complexity of the algorithm.

For each node $v \in V$ and index $j \in \set{1,\ldots,p}$,
let $T_j^v = S_{g(v)_j}$.
(That is, if $g(v) = (i_1,\ldots,i_p)$, then $T_1^v = S_{i_1}, \ldots, T_p^v = S_{i_p}$.)
Also, let $T^v \coloneq \bigcup_{j = 1}^p T_j^v$.
Node $v$ is responsible for listing all $p$-cliques $(u_1,\ldots,u_p) \in \bigtimes(T_1^v, \ldots,  T_p^v)$.
To do so, node $v$ needs to learn all edges
in $E(T_1^v, \ldots,  T_p^v)$;
there are at most $p(n^{1-1/p})^2 = O(n^{2-2/p})$ such edges,
which, using Lenzen's routing scheme, can be collected in $O(n^{2-2/p-1}) = O(n^{1-2/p})$ rounds.
After collecting the edges in $E(T_1^v,\ldots, T_p^v)$, node $v$ locally enumerates all $p$-cliques it sees.

\subsection[$K_{p+1}$-Detection from $K_p$-Listing]{\boldmath $K_{p+1}$-Detection from $K_p$-Listing}

We start by showing how to use the $K_p$ listing algorithm for obtaining $K_{p+1}$-detection, stated as follows.

\begin{theorem}
	For any $p \geq 3$,
	the $K_{p+1}$-detection problem in the $\qclique$ can be solved in $\tilde{O}(n^{(1-1/p)/2 } + n^{1-2/p} )$ rounds.
	\label{thm:plus1}
\end{theorem}

\begin{proof}
First, we let each node $v \in V$  collect all edges in $E(T^v \times T^v)$, and list all cliques in $\bigtimes(T_1^v, \ldots, T_p^v)$, as described above.

Our goal now is for each node $v \in V$
to check whether there is a node $u \in V$ that
extends a $p$-clique found by node $v$ in the previous step into a $(p+1)$-clique.
In other words,
node $v$ searches for $(p+1)$-cliques in $\bigtimes(T_1^v, \ldots, T_p^v, V)$.

To speed up the search, we partition the nodes of $V$ into $n^{1-1/p}$ batches, $Q_1,\ldots,Q_{n^{1-1/p}}$,
each of size $n^{1/p}$.
We then use a Grover search (Lemma~\ref{lemma:qsearchv0}) to find an index $i$ such that $\bigtimes(T_1^v,\ldots, T_p^v, Q_i)$
contains a $(p+1)$-clique (or determine that there is no such $i$).

We denote by $A_i$ a query for determining whether $\bigtimes(T_1^v, \ldots, T_p^v, Q_i)$ contains a $(p+1)$-clique, and we now describe how each query is implemented.
In $A_i$, node $v$ needs to collect all edges in $E(T, Q_i)$; recall that $T = \bigcup_{j = 1} T_j^v$, and
its size is $|T| \leq p\cdot n^{1-1/p}$.
Since $|Q_i| = n^{1/p}$, the number of edges $v$ needs to learn about is $n^{1-1/p} \cdot n^{1/p} = n$,
and using Lenzen's routing scheme.
this can be done in a single round.
After learning the relevant edges, node $v$ locally searches for a $(p+1)$-clique contained in the edges it has learned;
it outputs ``yes'' if and only if it finds one.

The running time of the resulting algorithm is as follows. Executing the $K_p$-listing algorithm from \cite{DLP12} requires $O(n^{1-2/p})$ rounds. The Grover search requires $\sqrt{ n^{1-1/p } }$ queries, each requiring a single round.
Using Lemma~\ref{lemma:qsearchv0}, the overall running time is therefore $\tilde{O}(n^{(1-p)/2})$
for the quantum search, in addition to the $O(n^{1-2/p})$ rounds required for $K_p$-listing.
Specifically, for $K_4$-detection ($p = 3$),
the total running time we obtain is $\tilde{O}(n^{1-2/3} + n^{1/2 - 1/6}) = \tilde{O}(n^{1/3})$.
\end{proof}

Theorem \ref{thm:plus1} gives, combined with the algorithm of Section \ref{sec:tri}, the statement of Theorem \ref{th:clique-clique} in the introduction.

\subsection[$K_{p+t}$-detection from $K_p$-listing]{\boldmath  $K_{p+t}$-Detection from $K_p$-Listing}

Now, we prove the following theorem, which shows how to obtain $K_{p+t}$-detection from $K_p$-listing.
\begin{theorem}
	Given a graph $G = (V, E)$ and values $p$, $t$ such that 
	$t \leq 1 + \log \left(p -1 \right)$,
	it is possible to detect if $G$ contains an instance of $K_{p + t}$ in
	$\tilde{O}( n^{ (1 - 1/p)(1 - 1/2^{t}) } + n^{1-2/p} )$
	rounds of communication in the \qclique, w.h.p. 
	\label{thm:newCCextend}
\end{theorem}

Optimizing our results over all choices of $p,t$ yields the following:
\begin{corollary}
	\label{cor:newCCextend}
	The time required to solve $K_{p+t}$-detection in $\qclique$ is
		\[\tilde{O}\left( \min_{p,t: t\leq 1+\log\left( p-1 \right) } \left(n^{ (1 - 1/p)(1 - 1/2^{t}) } + n^{1-2/p} \right)\right).\]
\end{corollary}

\begin{proof}[Proof of Theorem~\ref{thm:newCCextend}]
$\,$\\
\textbf{High-level overview.} 
We again first let each node $v \in V$  collect all edges in $E(T^v \times T^v)$,
and list all cliques in $\bigtimes(T_1^v, \ldots, T_p^v)$.
Now we want to solve $(p+t)$-clique detection, having already listed all $p$-cliques using the algorithm of \cite{DLP12}.

To do this, we generalize the +1 extension from the previous section using a \emph{recursive} search procedure,
$\proc{check}_k(P_1,\ldots,P_k,F)$ where $k \leq p+t$.
The procedure is given $k$ sets of nodes, $P_1,\ldots,P_k \subseteq V$,
and the set $F$ of all edges in $E(P_1,\ldots,P_k)$;
it recursively checks whether there is a $k$-clique $(u_1,\ldots,u_k) \in \bigtimes(P_1, \ldots, P_k)$
that can be extended into a $(p+t)$-clique.
The recursion begins at $k = p$, with the sets $P_1 = T_1^v,\ldots,P_p = T_p^v$, and the edges that node $v$ collected
in the pre-processing stage;
the recursion terminates at $k = p+t$,
where node $v$ enumerates all $(p+t)$-tuples $(u_1,\ldots,u_{p+t}) \in P_1 \times \ldots \times P_{p+t}$
and checks using the edges of $F$ whether one of them forms a $(p+t)$-clique.
Next, we explain how each internal level of the recursion is implemented.
Fix $t$ predetermined partitions of the nodes of $V$, where the $i$-th partition is given by $X_i^1,\ldots,X_i^{n^{r_i}}$
for a parameter $r_i \in (0,1)$, and all subsets in each partition have the same size:
$|X_i^1| = \ldots = |X_i^{n^{r_i}}| = n^{1-r_i}$.
(We assume again for simplicity that $n^{r_i}$ is an integer.)
For $i \in \set{0,\ldots,t-1}$, when $\proc{check}_{p+i}(P_1,\ldots,P_{p+i}, F)$
is called,
node $v$ uses a Grover search to check if there is some set $X_{i+1}^j$
such that $\proc{check}_{p+i}(P_1,\ldots,P_{p+i}, X_{i+1}^j, F')$ returns true,
where 
$F' = E(P_1,\ldots,P_{p+i}, X_{i+1}^j)$.
Each query $A_j$ in the Grover search takes the index $j \in [n^{r_i}]$ as input,
learns all the edges in $F' \setminus F$ (the edges of $F$ are already known),
and calls $\proc{check}_{p+i+1}(P_1,\ldots,P_{p+i}, X_{i+1}^j, F')$.\vspace{2mm}

\noindent\textbf{The formal algorithm.} 
Fix $p \geq 2$ and $t \geq 1$ satisfying the constraint.
Also, for each $i = 1,\ldots,t$,
fix a partition $X_i^1,\ldots,X_i^{n^{r_i}}$ of $V$,
where
$$r_i = (1-1/p)/2^{t-i},$$
and $|X_i^1| = \ldots = |X_i^{n^{r_i}}| = n^{1-r_i}$.
The partitions are arbitrary, but fixed in advance, so no communication is necessary to compute them.
Note that $r_i \geq 1/p$ for each $i = 1,\ldots,t$.

Our algorithm is a depth-$t$ nested quantum search, over the search space 
$[ n^{r_1} ] \times \ldots \times [ n^{r_t}]$,
with the goal function $f : [ n^{r_1} ] \times \ldots \times [ n^{r_t} ] \rightarrow \set{0,1}$
such that 
	$f(i_1,\ldots,i_t) = 1$
	iff
	there exists a $p$-clique $v_1,\ldots,v_p \in V^p$,
	and there exist $u_1 \in X_1^{i_1},\ldots,u_t \in X_t^{i_t}$,
	such that 
	$v_1,\ldots,v_p,u_1,\ldots,u_t$ is a $(p+t)$-clique.

	Before the quantum search begins, the nodes list all $p$-cliques in the graph;
	each node $v$ learns all edges in $E(T_1^v,\ldots,T_p^v)$ as described in the beginning of the section.

	The setup algorithms $\mathcal{S}_1,\ldots,\mathcal{S}_{t-1}$
	are as follows: in $\mathcal{S}_i$,
	the leader disseminates the current partial search query $(j_1,\ldots,j_i) \in [n^{r_1}] \times \ldots \times [n^{r_i}]$.
	The nodes use Lenzen's routing scheme (Lemma~\ref{lem:lenzen_routing})
	so that each $v \in V$ learns the edges
	$E(T_1^v \cup \ldots \cup T_p^v \cup X_1^{j_1} \cup \ldots \cup X_{i-1}^{j_{i-1}}, X_i^{j_i})$.
	The running time $s_i$ of $\mathcal{S}_i$ is
	$O(n^{1-1/p-r_i})$:
	we have $|T_1^v \cup \ldots \cup T_p^v| = O(n^{1-1/p})$ (treating $p$ as a constant),
	$|X_1^{j_1} \cup \ldots \cup X_{i-1}^{j_{i-1}}| = O(n^{1-1/p})$ 
	(by our assumption that $r_i \geq 1/p$ and treating $t$ as a constant), and 
	$|X_i^{j_i}| = O(n^{1-r_i})$.
	By Lemma~\ref{lem:lenzen_routing},
	the information can be routed in $O(n^{1 - 1/p + 1 - r_i - 1}) = O(n^{1-1/p - r_i})$ rounds.

	The final classical evaluation procedure, $\mathcal{C}$,
	is as follows: the leader disseminates the query $(j_1,\ldots,j_t) \in [n^{r_1}] \times \ldots \times [n^{r_t}]$
	to all nodes,
	and the nodes use Lenzen's routing scheme so that each $v \in V$ learns the edges
	$E(T_1^v \cup \ldots \cup T_p^v \cup X_1^{j_1} \cup \ldots \cup X_{t-1}^{j_{t-1}}, X_t^{j_t})$.
	Together with the edges in $\bigcup_{i < t} \mathsf{setup}_i^t(j_1,\ldots,j_i)$,
	each node $v$ now knows $E(T_1^v,\ldots,T_p^v,X_1^{j_1},\ldots,X_t^{j_t})$.
	Node $v$ now checks whether it sees a $(p+t)$-clique, and informs the leader.
	The running time of $\mathcal{C}$
	is $O(n^{1-1/p-r_t})$ (as above).

	By Lemma~\ref{lemma:qsearch}, the overall running time of the quantum search is given by
	\begin{align*}
		\!\!\!\!\!\!\tilde{O}\!\left( n^{r_1/2} \!\left( n^{1-1/p-r_1} \!+\! n^{r_2/2} \!\cdot \!\left( n^{1-1/p-r_2} \!+\!  \ldots \!+\!  n^{r_{t-1}/2} \left(n^{1-1/p-r_{t-1}} \!+\! n^{r_t/2} \cdot n^{1-1/p-r_t}  \right)\! \right)\!\right) \!\right)\!.
	\end{align*}
	Denote by $g(i)$ the running time of the innermost $t - i + 1$ levels of the search, starting
	from level $i$ up to level $t$:
	\begin{align*}
		g(i) = 
		\tilde{O}\left( n^{r_i/2} \left( n^{1-1/p-r_i} + \ldots +  n^{r_{t-1}/2} \left(n^{1-1/p-r_{t-1}} + n^{r_t/2} \cdot n^{1-1/p-r_t}  \right)  \right) \right).
	\end{align*}
	We claim, by backwards induction on $i$, that
	\begin{equation}
		g(i) = \tilde{O}( n^{(1-1/p)(1 - 1/2^{t-i+1})} ).
		\label{eq:gi}
	\end{equation}
	For the base case, $i = t$, we have 
	\begin{align*}
		g(i)
		= \tilde{O}\left( n^{r_t/2} \cdot n^{1-1/p-r_t} \right)
		= \tilde{O}\left( n^{1-1/p-r_t/2} \right)
		=
		\tilde{O}\left( n^{1-1/p-(1-1/p)/2} \right)
		=
		\tilde{O}\left( n^{(1-1/p)/2} \right),
	\end{align*}
	and indeed this matches~\eqref{eq:gi}.
	For the induction step, suppose the claim holds for $i$, and consider $g(i-1)$:
	\begin{align*}
		g(i-1)
		&=
		\tilde{O}\left( n^{r_{i-1}/2} \left( n^{1-1/p-r_{i-1}} + g(i) \right) \right)
		\\
		&=
		\tilde{O}\left( n^{(1-1/p)/2^{t-i+1}/2} \left( n^{1-1/p-(1-1/p)/2^{t-i+1}} + n^{(1-1/p)(1 - 1/2^{t-i+1})} \right) \right)
		\\
		&=
		\tilde{O}\left( n^{(1-1/p)/2^{t-i+2}} \cdot 2n^{(1-1/p)(1 - 1/2^{t-i+1})} \right)
		\\
		&=
		\tilde{O}\left( n^{(1-1/p)(1 - 1/2^{t-i+2})} \right),
	\end{align*}
	again matching~\eqref{eq:gi}.

	We thus obtain that the running time of the entire quantum search is
	\begin{equation*}
		g(1) = \tilde{O}\left( n^{(1-1/p)(1 - 1/2^{t})} \right),
	\end{equation*}
	and together with the complexity of listing the initial $p$-clique, we get the claimed complexity.
\end{proof}

\noindent\textbf{A note on the restriction on \boldmath{$p,t$}.}
Recall that we assumed throughout that $1 - r_i \leq 1 - 1/p$ for each $i = 1,\ldots,t$,
so that the sets of nodes we handle at each step never exceed the size of the sets in the $p$-clique-listing
pre-processing step ($n^{1 - 1/p}$).
In other words, we must have
	$r_i \geq \frac{1}{p}$.
The value of $r_i$ decreases with $i$, so it suffices to require that $r_1 \geq 1/p$.
We now see that not every choice of $p,t$ respects this condition:
since we set $r_1 = \frac{1}{2^{t-1}}\left( 1 - \frac{1}{p} \right)$,
we require that $p,t$ satisfy
	$t \leq 1 + \log \left(p -1 \right)$.

For example, suppose we want to start by lising all triangles ($p = 3$), and extend to $K_5$ (i.e., $t = 2$).
This is possible, since $
	2 \leq 1 + \log\left( 3 - 1 \right).
	$
However, extending from triangles to $K_6$ using the approach described here is not possible,
because if we take $p = 3$ and $t = 3$,
we get $
	3 \not \leq 1 + \log \left( 3 - 1 \right).
	$

Although we think of $p,t$ as constants, we observe that when $p$ is large and $t \approx \log(p)$,
the overall running time of the quantum part of our scheme is $\approx n^{1-2/p}$.
Thus, the cost of extending from $p$-cliques to $(p+t)$-cliques roughly matches the cost of the classical $p$-clique listing,
which is $n^{1-2/p}$, and we get a quantum algorithm for $K_{p+t}$-detection that roughly matches the cost of
classical $K_p$-listing.


\section{Circuit-Complexity Barrier to Proving an {\boldmath $\Omega(n^{3/5+\alpha})$}-Round Lower Bound for Clique Detection in the  \congest Model}\label{sec:barrier}
In this section, we show that for any $\alpha > 0$, a lower bound of the form $\Omega(n^{3/5+\alpha})$ for $K_p$-detection in (non-quantum) $\CONGEST$ would imply strong circuit complexity results, far beyond the current state-of-the-art. As mentioned in the introduction, this barrier also applies to the $\qcongest$ model.

Given a constant integer $\alpha > 0$, let $\mathcal{F}_{\alpha}$ be the family of Boolean circuits $F$ where:
\begin{itemize}
	\item $F$ has $M\log{M}$ input wires for some integer $M$, which we interpret	as encoding a graph $\bar{G}$ on $O(M)$ edges.
	\item All gates in $F$ have constant fan-in and fan-out.
	\item $F$ has depth at most $R = M^{\alpha}$.
	\item $F$ has at most $M^{1+\alpha}$ wires in total.
\end{itemize}

\begin{theorem}
	\label{thm:clique_barrier}
	If $K_p$-detection has a round complexity of $\Omega(n^{3/5+\alpha})$ for some constant $\alpha > 0$,
	then there is no circuit family contained in $\mathcal{F}_{\alpha}$ that solves $K_p$-detection.
\end{theorem}

In other words, a lower bound of the form $\Omega(n^{3/5+\alpha})$ for any constant $\alpha > 0$
rules out the existence of a circuit with polynomial depth and super-linear size (in terms of wires)
for an explicit problem, $K_p$-detection. Such a lower bound would be a major breakthrough in circuit complexity (see~\cite{EFFKO19} for discussion). This gives the statement of Theorem~\ref{thm:Kp_obstacle} in the introduction.

The proof of Theorem~\ref{thm:clique_barrier} is essentially a reduction from $K_p$-detection
in general graphs to $K_p$-detection in high-conductance graphs. It was shown in~\cite{EFFKO19} that in high-conductance graphs, we can efficiently \emph{simulate} a circuit from the family $\mathcal{F}_{\alpha}$ of certain size, so given such a circuit for $K_p$-detection, we can use it to solve the distributed $K_p$-detection problem in high-conductance graphs. To reduce from general graphs to high-conductance graphs, we use a similar approach to the clique detection algorithm of~\cite{Censor+SODA21}: we first compute an expander-decomposition procedure from to partition the graph into high-conductance clusters with few inter-cluster edges, and have the nodes of the clusters learn the relevant sets of edges in the clusters' neighborhood, so that we either find a $K_p$-copy or can remove most of edges of the graph without removing a $K_p$-copy, and recurse on the remaining edges. 


\subsection{Reducing {\boldmath $K_p$}-Freeness From General Graphs to High-Conductance Graphs}
To reduce from general graphs to high-conductance components, we use several results from~\cite{Censor+SODA21}. We require slightly different parameters than the ones used in~\cite{Censor+SODA21}, so for the sake of completeness,
we re-analyze these procedure for in Appendix~\ref{sec:subprocedures}.

For a set of vertices $A$, we denote by $N(A) = (\bigcup_{v \in A} N(v))\setminus A$ the set of neighbors of $A$.

Assume that we have a procedure $\mathcal{A}$ which solves $K_p$-detection on high conductance graphs, where nodes may have some additional input edges, which are not communication edges. In this subsection, we show how to solve $K_p$-detection in general graphs using $\mathcal{A}$, and in the next subsection we show how to obtain an efficient procedure $\mathcal{A}$ assuming the existence of a circuit family $\mathcal{F}_{\alpha}$ that solves $K_p$-detection. The reduction now proceeds as follows. 

First, the network runs the decomposition of Lemma~\ref{lem:decomp_soda} to obtain a partition of the edges to set $E_m,E_r$ with the described properties. 

Next, we show that in parallel for all small clusters $C$, the network can can determine efficiently whether there is a $K_p$-copy containing a node of $C$. This is again a variation on a claim from~\cite{Censor+SODA21}.

\begin{lemma}[\cite{Censor+SODA21}, Lemma 4.2]\label{lem:handle_small_clusters}
	There is a procedure that terminates after $O(n^{3/5})$ rounds in which the network can determine whether there is a $K_p$-copy which at least one of its nodes is contained in a cluster of size $|V_C| \leq n^{3/5}$.
\end{lemma}

Following this, by Lemma~\ref{lem:insert_edges_into_clusters} there is an $\tilde{O}(n^{3/5})$-round procedure in which for each cluster $C$, and each vertex $v \in V_C$ in the cluster learns the neighborhood induced by their neighbors outside~$C$.

Therefore, we turn our attention to large clusters. For a cluster $C$, let $S^*_C = \{u \in N(C) \mid \deg_{V_C}(u) \leq \deg_{V\setminus V_C}(u)/n^{3/5} \}$ be the set of nodes
$v$ adjacent to the cluster,
such that the number of neighbors $v$ has in $C$ is larger by a factor of at least $n^{3/5}$ than the number of neighbors $v$
has outside $C$. Let $S_C = \{u \in V_C \mid \deg_{S^*_C(v)} \geq n^{3/5}\}$ be the set of nodes in $C$ that have at least $n^{3/5}$ neighbors in $S^*_C$. Informally, we think of $S_C$ as \emph{bad} cluster nodes, and $S^*_C$ as bad nodes neighboring the cluster.

We claim that each large cluster can efficiently ``pull in'' all the ``non-bad'' edges in its vicinity that it needs to check for copies of $K_p$ that involve a node in $C \setminus S_C$:

\begin{lemma}[Implicit in \cite{Censor+SODA21}; Lemma~\ref{lem:insert_edges_into_clusters_app} in Appendix~\ref{sec:subprocedures}]\label{lem:insert_edges_into_clusters}
	In $\tilde{O}(n^{3/5})$ rounds, the network can for each cluster $C$ with $|V_C| \geq n^{3/5}$ partition $E(N(C)\setminus S^*_C,N(C))$ into sets $\{E_{v,C}\}_{v \in C}$, each of size at most $\tilde{O}(n^{3/5}\deg{v})$, and have each node $v \in V_C$ learn the edges of $E_{v,C}$, and in addition, each node $v \in C \setminus S_C$ learns the set $E'_{v,C} = E(N(v) \cap S^*_C,N(v) \cap S^*_C)$. 
\end{lemma}

Next, each cluster runs the $K_p$-detection procedure $\mathcal{A}$ with the edges of the cluster as inputs, and where every node is given as additional input the set of edges $E_{v,C} \cup E'_{v,C}$.

We note that by Lemma~\ref{lem:insert_edges_into_clusters}, the only edges from the vicinity of a cluster $C$ not contained in $\bigcup_{v \in C} E_{v,C} \cup \bigcup_{v \in C\setminus S_C} E'_{v,C}$ are edges in $S^*_C \times S^*_C$ whose endpoints do not share a neighbor in $C \setminus S_C$. Therefore, if no $K_p$-copy was found in the clusters, we know that any $K_p$-copy must be contained in the edge set
\[
E_{\var{recurse}} = E_r \cup \bigcup_{C:|V_C| \geq n^{3/5}}  E(S_C,S_C), 
\]

In the following lemma, we note that the total number of edges in between any two bad cluster nodes of in all of the clusters is not large.

\begin{lemma}[Implicit in \cite{Censor+SODA21}; Lemma~\ref{lem:size_of_sc_app} in Appendix~\ref{sec:subprocedures}]
	We have $\bigcup_{C:|V_C| \geq n^{3/5}}  |E(S_C,S_C)| \leq \epsilon m$.
\end{lemma}

Therefore the set $E_{\var{recurse}}$ is of size at most $2\epsilon m$, and we recurse on this set. If a node found no $K_p$-copy in all recursion steps, it accepts.  

\subsection{Simulating Circuits On High-Conductance Graphs}

Next, we use the theorem presented in \cite{EFFKO19} for relating lower bounds in high conductance networks to circuit complexity.

\begin{lemma}[\cite{EFFKO19}] \label{lem:circ-to-cong}
	Let $U$  be a graph $U=(V,E)$ with $|V|=\bar{n}$ vertices and $|E|=\bar{n}^{1+\rho}$ edges, with mixing time $\tau_{\textrm{mix}}$. Suppose that $f:\{0,1\}^{c\bar{n}^{1+\rho}\log{n}}\rightarrow \{0,1\}$ for some integer $c > 1$ is computed by a circuit $\mathcal{C}$ of depth $R$, comprising of gates with constant fan-in and fan-out, and at most $O(c \cdot s\cdot \bar{n}^{1+\rho}\log{\bar{n}})$ wires. Then for any input partition that assigns to each vertex in $U$ no more than $c \deg(v)\log{n}$ input wires, there is an $O(R \cdot c \cdot s\cdot \tau_{\textrm{mix}}\cdot 2^{O(\sqrt{\log \bar{n}\log\log \bar{n}})})$-round protocol in the $\CONGEST$ model on $U$
	that computes $f$ under the input partition. 
\end{lemma}

Taking $c = \tilde{\Theta}(n^{3/5})$, $R=n^{\alpha/2}, s = \Theta(n^{\alpha/2})$ and $f$ to be the function that given the encoding of edges outputs whether the graph on these edges is $K_p$-free, we get that given a constant fan-in fan-out gate circuit with $O(s\cdot c \cdot n^{1+\delta}\log{n})$ wires and depth $R$ that solves $f$, we could solve $K_p$-detection in \congest in $O(n^{3/5+\alpha+o(1)})$ rounds, since each node has indeed $\tilde{O}(n^{3/5}\deg_C(v))$ input edges to encode, we can encode its input with at most $\tilde{O}(n^{3/5}\deg_C(v))$ input wires. Therefore, if the round complexity of $K_p$-detection in \congest is $\Omega(n^{3/5+\alpha})$ for some constant $\alpha > 0$, there are circuit family in $\mathcal{F}_{\alpha'}$ for some constant $\alpha' > 0$ which solves $K_p$-detection, which implies Theorem~\ref{thm:clique_barrier}.

\bibliographystyle{plain}
\bibliography{References}

\begin{thebibliography}{10}

\bibitem{ACHKL20}
Amir Abboud, Keren Censor{-}Hillel, Seri Khoury, and Christoph Lenzen.
\newblock Fooling views: a new lower bound technique for distributed
  computations under congestion.
\newblock {\em Distributed Computing}, 33(6):545--559, 2020.

\bibitem{Arfaoui+14}
Heger Arfaoui and Pierre Fraigniaud.
\newblock What can be computed without communications?
\newblock {\em {SIGACT} News}, 45(3):82--104, 2014.

\bibitem{Ben-Or+STOC05}
Michael Ben{-}Or and Avinatan Hassidim.
\newblock Fast quantum byzantine agreement.
\newblock In {\em Proceedings of the 37th Annual {ACM} Symposium on Theory of
  Computing (STOC 2005)}, pages 481--485, 2005.

\bibitem{Censor+SODA21}
Keren Censor{-}Hillel, Yi-Jun Chang, Fran{\c{c}}ois~Le Gall, and Dean
  Leitersdorf.
\newblock Tight distributed listing of cliques.
\newblock In {\em Proceedings of the Twenty-Ninth Annual {ACM-SIAM} Symposium
  on Discrete Algorithms ({SODA} 2021)}, pages 2878--2891, 2021.

\bibitem{CHFGLGDR20}
Keren Censor-Hillel, Orr Fischer, Tzlil Gonen, Fran{\c{c}}ois Le~Gall, Dean
  Leitersdorf, and Rotem Oshman.
\newblock Fast distributed algorithms for girth, cycles and small subgraphs.
\newblock In {\em Proceedings of the 34th International Symposium on
  Distributed Computing (DISC 2020)}, pages 33:1--33:17, 2020.

\bibitem{CHLGL20}
Keren Censor-Hillel, Fran{\c{c}}ois~Le Gall, and Dean Leitersdorf.
\newblock On distributed listing of cliques.
\newblock In {\em Proceedings of the 39th ACM Symposium on Principles of
  Distributed Computing (PODC 2020)}, pages 474--482, 2020.

\bibitem{CHKKLPJ19}
Keren Censor{-}Hillel, Petteri Kaski, Janne~H. Korhonen, Christoph Lenzen, Ami
  Paz, and Jukka Suomela.
\newblock Algebraic methods in the congested clique.
\newblock {\em Distributed Computing}, 32(6):461--478, 2019.

\bibitem{CPZ19}
Yi{-}Jun Chang, Seth Pettie, and Hengjie Zhang.
\newblock Distributed triangle detection via expander decomposition.
\newblock In {\em Proceedings of the Thirtieth Annual {ACM-SIAM} Symposium on
  Discrete Algorithms ({SODA} 2019)}, pages 821--840, 2019.

\bibitem{CS19}
Yi{-}Jun Chang and Thatchaphol Saranurak.
\newblock Improved distributed expander decomposition and nearly optimal
  triangle enumeration.
\newblock In {\em Proceedings of the 2019 {ACM} Symposium on Principles of
  Distributed Computing (PODC 2019)}, pages 66--73, 2019.

\bibitem{CK20}
Artur Czumaj and Christian Konrad.
\newblock Detecting cliques in {CONGEST} networks.
\newblock {\em Distributed Computing}, 33(6):533--543, 2020.

\bibitem{Denchev+08}
Vasil~S. Denchev and Gopal Pandurangan.
\newblock Distributed quantum computing: a new frontier in distributed systems
  or science fiction?
\newblock {\em SIGACT News}, 39(3):77--95, 2008.

\bibitem{DLP12}
Danny Dolev, Christoph Lenzen, and Shir Peled.
\newblock ``tri, tri again'': Finding triangles and small subgraphs in a
  distributed setting.
\newblock In {\em Proceedings of the 26th International Symposium on
  Distributed Computing (DISC 2012)}, pages 195--209, 2012.

\bibitem{EFFKO19}
Talya Eden, Nimrod Fiat, Orr Fischer, Fabian Kuhn, and Rotem Oshman.
\newblock Sublinear-time distributed algorithms for detecting small cliques and
  even cycles.
\newblock In {\em Proceedings of the 33rd International Symposium on
  Distributed Computing (DISC 2019)}, pages 15:1--15:16, 2019.

\bibitem{EFFKO_note}
Talya Eden, Nimrod Fiat, Orr Fischer, Fabian Kuhn, and Rotem Oshman.
\newblock Sublinear-time distributed algorithms for detecting small cliques and
  even cycles (journal version).
\newblock To appear in Distributed Computing, 2021.

\bibitem{EKNP14}
Michael Elkin, Hartmut Klauck, Danupon Nanongkai, and Gopal Pandurangan.
\newblock Can quantum communication speed up distributed computation?
\newblock In {\em Proceedings of the 2014 ACM Symposium on Principles of
  Distributed Computing (PODC 2014)}, pages 166--175, 2014.

\bibitem{FGKO18}
Orr Fischer, Tzlil Gonen, Fabian Kuhn, and Rotem Oshman.
\newblock Possibilities and impossibilities for distributed subgraph detection.
\newblock In {\em Proceedings of the 30th on Symposium on Parallelism in
  Algorithms and Architectures ({SPAA} 2018)}, pages 153--162, 2018.

\bibitem{FHW12}
Silvio Frischknecht, Stephan Holzer, and Roger Wattenhofer.
\newblock Networks cannot compute their diameter in sublinear time.
\newblock In {\em Proceedings of the Twenty-Third Annual {ACM-SIAM} Symposium
  on Discrete Algorithms (SODA 2012)}, pages 1150--1162, 2012.

\bibitem{GKM09}
Cyril Gavoille, Adrian Kosowski, and Marcin Markiewicz.
\newblock What can be observed locally?
\newblock In {\em Proceedings of the 23rd International Symposium on
  Distributed Computing (DISC 2009)}, pages 243--257, 2009.

\bibitem{Ghaffari+PODC17}
Mohsen Ghaffari, Fabian Kuhn, and Hsin-Hao Su.
\newblock Distributed {MST} and routing in almost mixing time.
\newblock In {\em Proceedings of the 2017 ACM Symposium on Principles of
  Distributed Computing (PODC 2017)}, pages 131--140, 2017.

\bibitem{Grover96}
Lov~K. Grover.
\newblock A fast quantum mechanical algorithm for database search.
\newblock In {\em Proceedings of the Twenty-Eighth Annual {ACM} Symposium on
  the Theory of Computing (STOC 1996)}, pages 212--219, 1996.

\bibitem{Izumi+PODC19}
Taisuke Izumi and Fran{\c{c}}ois~Le Gall.
\newblock Quantum distributed algorithm for the all-pairs shortest path problem
  in the {CONGEST-CLIQUE} model.
\newblock In {\em Proceedings of the 2019 {ACM} Symposium on Principles of
  Distributed Computing (PODC 2019)}, pages 84--93, 2019.

\bibitem{ILM20}
Taisuke Izumi, Fran{\c{c}}ois~Le Gall, and Fr{\'{e}}d{\'{e}}ric Magniez.
\newblock Quantum distributed algorithm for triangle finding in the {CONGEST}
  model.
\newblock In {\em Proceedings of the 37th International Symposium on
  Theoretical Aspects of Computer Science (STACS 2020)}, pages 23:1--23:13,
  2020.

\bibitem{ILG2017}
Taisuke Izumi and Fran\c{c}ois Le~Gall.
\newblock Triangle finding and listing in {CONGEST} networks.
\newblock In {\em Proceedings of the 2017 ACM Symposium on Principles of
  Distributed Computing (PODC 2017)}, pages 381--389, 2017.

\bibitem{JS89}
Mark Jerrum and Alistair Sinclair.
\newblock Approximating the permanent.
\newblock {\em SIAM journal on computing}, 18(6):1149--1178, 1989.

\bibitem{LeGall+PODC18}
Fran{\c{c}}ois {Le Gall} and Fr{\'{e}}d{\'{e}}ric Magniez.
\newblock Sublinear-time quantum computation of the diameter in {CONGEST}
  networks.
\newblock In {\em Proceedings of the 2018 {ACM} Symposium on Principles of
  Distributed Computing (PODC 2018)}, pages 337--346, 2018.

\bibitem{LNR19}
Fran{\c{c}}ois {Le Gall}, Harumichi Nishimura, and Ansis Rosmanis.
\newblock Quantum advantage for the {LOCAL} model in distributed computing.
\newblock In {\em Proceedings of the International Symposium on Theoretical
  Aspects of Computer Science (STACS 2019)}, pages 49:1--49:14, 2019.

\bibitem{Lenzen13}
Christoph Lenzen.
\newblock Optimal deterministic routing and sorting on the congested clique.
\newblock In {\em Proceedings of the 2013 {ACM} Symposium on Principles of
  Distributed Computing ({PODC} 2013)}, pages 42--50, 2013.

\bibitem{PS18}
Gopal Pandurangan, Peter Robinson, and Michele Scquizzato.
\newblock On the distributed complexity of large-scale graph computations.
\newblock In {\em Proceedings of the 30th on Symposium on Parallelism in
  Algorithms and Architectures (SPAA 2018)}, pages 405--414, 2018.

\bibitem{SJSS93}
Jeanette~P. Schmidt, Alan Siegel, and Aravind Srinivasan.
\newblock {Chernoff-Hoeffding} bounds for applications with limited
  independence.
\newblock In {\em Proceedings of the Fourth Annual ACM-SIAM Symposium on
  Discrete Algorithms (SODA 1993)}, pages 331--340, 1993.

\bibitem{Tani+12}
Seiichiro Tani, Hirotada Kobayashi, and Keiji Matsumoto.
\newblock Exact quantum algorithms for the leader election problem.
\newblock {\em {ACM} Transactions on Computation Theory}, 4(1):1:1--1:24, 2012.

\end{thebibliography}

\appendix
\section{Re-Analysis of Sub-procedures of~\cite{Censor+SODA21}}
\label{sec:subprocedures}
For completeness, we repeat the analysis done in~\cite{Censor+SODA21} for the sub-procedure which partitions the graph into clusters and has all but a few nodes in most clusters learn its neighborhood,  with different parameters.


\begin{lemma}[based on \cite{CS19}, Theorem 1 and optimal cliques]\label{lem:decomp_soda_app}
	Let $\delta > 0$ such that $m = n^{1+\delta}$. For any $\epsilon \in (0,1)$, and constant $\gamma \in (0,1)$, there is a constant $\alpha_\gamma > 0$ dependent only on $\gamma$, such that a decomposition can be constructed in $\tilde{O}(n^\gamma)$ rounds, with high probability, in which the edges of the graph are partitioned into two sets, $E_m,E_r$, that satisfy the following conditions:
	\begin{enumerate}
		\item Each connected component (cluster) $C$ of $E_m$ has conductance $\Phi(C)\geq (\epsilon/\log{n})^{\alpha_\gamma}$, and has average degree at least $ \epsilon n^\delta$.
		\item For any cluster $C$ and node $v \in V_C$, $\deg_{V_C}(v) \geq (\epsilon/\log{n})^{\alpha_\gamma}\deg_{V \setminus V_C}(v)$.
		\item \label{prop_er_size} $E_r \leq \epsilon m$. 
	\end{enumerate}
\end{lemma}

%

\begin{proof}
	
	The network runs the expander decomposition of \cite{CS19} with parameters $(\epsilon/2,1/\polylog(n))$ to partition the graph into sets $E'_m,E'_r$ such that each connected component $C$ of $E'_m$ has conductance $\Phi(C)\geq (\epsilon/\log{n})^{\alpha_\gamma}$, for each $v \in V_C$, $\deg_{V_C}(v) \geq (\epsilon/\log{n})^{\alpha_\gamma}\deg_{V \setminus V_C}(v)$, and $E'_r \leq \epsilon m$, in $O(n^{\gamma})$ rounds. The network marks all edges in $E'_r$ to be in $E_r$. In parallel, each cluster $C$ computes a BFS spanning tree on the cluster, and counts the number of edges and nodes in $C$. If the average degree in $C$ is less than $\epsilon n^\delta$, the network adds all edges of $C$ into $E_r$. The total number of edges in clusters of $E'_m$ with average degree $2|E_C|/|V_C| \leq \epsilon n^\delta$ is at most $n \cdot (\epsilon n^\delta)/2 = \epsilon n^{1+\delta}/2 = (\epsilon/2) m$. The network marks the remaining edges as $E_m$. We note that indeed $E_r \leq (\epsilon/2+\epsilon/2)m = \epsilon m$. Each remaining cluster $C$ that was not removed to $E_r$ has average degree at least $\epsilon n^\delta$, and by the construction of \cite{CS19}, has $\Phi(C)\geq (\epsilon/\log{n})^{\alpha_\gamma}$ and for each $v \in V_C$, $\deg_{V_C}(v) \geq (\epsilon/\log{n})^{\alpha_\gamma}\deg_{V \setminus V_C}(v)$.
\end{proof}

For a cluster $C$, let $N(C) = \{v \mid \exists_{u \in C} (u,v) \in E'\}$ be the nodes outside the cluster which have a neighbor in $C$,. Let $\beta = 3/5$, $S^*_C = \{u \in N(C) \mid \deg_{V_C}(u) \leq \deg_{V\setminus V_C}(u)/n^\beta \}$, and let $S_C = \{u \in V_C \mid \deg_{S^*_C(v)} \geq n^\beta\}$.

%
%

%

\begin{lemma}[\cite{Censor+SODA21}]\label{lem:size_of_sc_app}
The size of $\bigcup_{C:|V_C| \geq n^{3/5}}  E(S_C,S_C)$ is at most $\epsilon m$.
\end{lemma}
\begin{proof}
We bound from above the size of $S_C$ with regards to $m$:
\begin{equation*}
2m \geq \sum_{u \in S^*_C}\deg_{V \setminus V_C}(u) \geq n^\beta \sum_{u \in S^*_C}\deg_{S^*_C}(u) \geq n^\beta \sum_{v \in S_C} \deg_{S^*_C}(u) \geq n^{2\beta} |S_C|.
\end{equation*}
Therefore,
\begin{equation}
|S_C| \leq 2m/n^{2\beta}
\end{equation}
and 
\begin{equation}
E(S_C,S_C) \leq |S_C|^2 \leq 2m/n^{4\beta} \leq 2n^2/n^{4\beta}.
\end{equation}

As the clusters are vertex-disjoint, there are at most $O(n^{1-\beta})$ clusters with more than $O(n^{\beta})$ vertices, therefore the total number of edges in $E(S_C,S_C)$ is bounded by $2cmn^{1-\beta}n^2/n^{4\beta} = 2cmn^{3-5\beta} \leq \epsilon m$.
\end{proof}

\begin{lemma}[\cite{Censor+SODA21}]\label{lem:insert_edges_into_clusters_app}
	In $\tilde{O}(n^{\beta})$ rounds, the network can for each cluster $C$ with $|V_C| \geq n^{\beta}$ partition $E(N(C)\setminus S^*_C,N(C))$ into sets $\{E_{v,C}\}_{v \in C}$, each of size at most $\tilde{O}(n^{\beta}\deg{v})$, and have each node $v \in V_C$ learn the edges of $E_{v,C}$, and in addition, each node $v$ in $C \setminus S_C$ learns the set $E'_{v,C} = E(N(v) \cap S^*_C,N(v) \cap S^*_C)$. 
\end{lemma}
\begin{proof}
		In parallel for all clusters $C$, each any node in $u \in N(C)\setminus S^*_C$ partitions its edges that have both endpoints in $V\setminus V_C$ into sets of size $\Theta(n^\beta)$, and sends each set to a unique neighbor in $C$. This is possible due to the fact that $\deg_{V_C}(u) \leq \deg_{V\setminus V_C}(u)/n^\beta$ and $n^\beta \geq n^{1-\beta}$. Following this, each node $v \in C\setminus S_c$ sends to each of its neighbors in $S^*_C$ the identifiers of its other neighbors in $S^*_C$ using $O(n^\beta)$ rounds. Then, each neighbor replies using $O(n^\beta)$ rounds, and sends to which of the nodes sent to it, it has an edge to. The returned edges are the set $E'_{v,C}$
\end{proof}

Finally, we cite another lemma which is used in Section~\ref{sec:barrier}.

\begin{lemma}[\cite{Censor+SODA21}, Lemma 4.2]\label{lem:handle_small_clusters_app}
	There is a procedure that terminates after $O(n^\beta)$ rounds network can determine whether there is a $K_p$-copy which at least one of its nodes is contained in a cluster of size $|V_C| \leq n^\beta$.
\end{lemma}  

\end{document}